\documentclass[a4paper,USenglish]{lipics-v2021}

\pdfoutput=1 %
\hideLIPIcs  %

\usepackage[T1]{fontenc}

\usepackage{mathtools}
\usepackage{booktabs}
\hypersetup{colorlinks=true,}

\usepackage{tikz}
\usetikzlibrary{shapes.arrows,chains,calc}
\usetikzlibrary{decorations.pathreplacing,matrix}

\newcommand{\Table}[1]{Table~\ref{tab:#1}}

\newcommand{\FigLabel}[1]{\label{fig:#1}}
\newcommand{\Figure}[1]{Figure~\ref{fig:#1}}

\newcommand{\SectLabel}[1]{\label{sect:#1}}
\newcommand{\Section}[1]{Sect.~\ref{sect:#1}}

\newcommand{\ThmLabel}[1]{\label{thm:#1}}
\newcommand{\Theorem}[1]{Theorem~\ref{thm:#1}}
\newcommand{\Thm}[1]{Thm~\ref{thm:#1}}
\newcommand{\LemLabel}[1]{\label{lem:#1}}
\newcommand{\Lemma}[1]{Lemma~\ref{lem:#1}}
\newcommand{\ConLabel}[1]{\label{con:#1}}
\newcommand{\Conjecture}[1]{Conjecture~\ref{con:#1}}
\newcommand{\DefLabel}[1]{\label{def:#1}}

\newcommand{\ObLabel}[1]{\label{ob:#1}}
\newcommand{\Observation}[1]{Observation~\ref{ob:#1}}
\newcommand{\PropLabel}[1]{\label{ob:#1}}
\newcommand{\Proposition}[1]{Proposition~\ref{ob:#1}}
\newcommand{\CorLabel}[1]{\label{cor:#1}}

\newcommand{\Cor}[1]{Cor.~\ref{cor:#1}}

\newcommand{\ie}{i.e.}
\newcommand{\eg}{e.g.}
\newcommand{\etal}{et~al.}

\newcommand{\Wilog}{W.l.o.g.}
\newcommand{\wilog}{w.l.o.g.}
\newcommand{\whp}{\textit{w.h.p.}}

\newcommand{\neworbetter}[1]{\textcolor{cyan!40!blue}{#1}}
\newcommand{\notneworbetter}[1]{#1}

\newcommand{\Set}[1]{\ensuremath{\{#1\}}}
\newcommand{\Card}[1]{\ensuremath{\left| #1 \right|}}
\newcommand{\Range}[2]{\ensuremath{\left[ #1, #2 \right]}}

\newcommand{\Sequence}[1]{\ensuremath{\left\langle #1 \right\rangle}}

\newcommand{\Edge}[2]{\ensuremath{\{#1,#2\}}}
\newcommand{\Neigh}{\ensuremath{N}}
\newcommand{\Degree}{\ensuremath{\mathrm{\deg}}}

\newcommand{\NeighEx}[1]{\ensuremath{N_{\bar #1}}}
\newcommand{\DegreeEx}[1]{\ensuremath{\mathrm{\deg}_{\bar #1}}}

\newcommand{\NP}{\ensuremath{\mathcal{NP}}}
\newcommand{\bigO}{\ensuremath{\mathcal{O}}}

\DeclareMathOperator{\polylog}{polylog}

\newcommand{\Alg}{\mathcal{A}}
\newcommand{\oneuMv}{\textsf{1-uMv}}
\newcommand{\OMv}{\textsf{OMv}}

\newcommand{\GraphletSet}{\ensuremath{\mathbb{P}}}
\newcommand{\Graphlet}{\ensuremath{\mathcal{P}}}
\newcommand{\PartFrac}{\ensuremath{\varepsilon}}
\newcommand{\High}{\ensuremath{\mathcal{H}}}
\newcommand{\Low}{\ensuremath{\mathcal{L}}}

\newcommand{\Count}{\mathbf{c}}
\newcommand{\ICount}{\mathbf{c_I}}

\newcommand{\DS}{\ensuremath{\mathcal{D}}}
\newcommand{\Aux}[1]{\texttt{#1}}

\newcommand{\FromThmCountAny}{${}^a$}%
\newcommand{\FromThmCountNoninduced}{${}^b$}%
\newcommand{\FromCorCountTriangles}{${}^c$}%
\newcommand{\FromThmCountNoninducedS}{${}^d$}%
\newcommand{\FromThmCountInduced}{${}^e$}%
\newcommand{\FromCorCountInduced}{${}^f$}%
\newcommand{\FromThmLowerbounds}{${}^g$}%
\newcommand{\FromEppsteinTri}{${}^\alpha$}
\newcommand{\FromEppsteinQuad}{${}^\beta$}
\newcommand{\FromDhulipala}{${}^\zeta$}
\newcommand{\PictoThreePath}{%
\begin{tikzpicture}[baseline=.3ex]
\draw (0,0) -- (0.5ex,1ex);
\draw (0.5ex,1ex) -- (1.5ex,1ex);
\draw (1.5ex,1ex) -- (2ex,2ex);
\end{tikzpicture}%
}

\newcommand{\PictoThreeCycle}{%
\begin{tikzpicture}[baseline=0ex]
\draw (0ex,0ex) -- (60:1.8ex) -- (1.8ex,0ex) -- cycle;
\end{tikzpicture}%
}

\newcommand{\PictoClaw}{%
\begin{tikzpicture}[baseline=.3ex]
\draw (1ex,2ex) -- (1ex,1ex);
\draw (0,0) -- (1ex,1ex);
\draw (2ex,0ex) -- (1ex,1ex);
\end{tikzpicture}%
}

\newcommand{\PictoPaw}{%
\begin{tikzpicture}[baseline=.3ex]
\draw (0,1ex) -- (1ex,1ex);
\draw (1ex,1ex) -- (2ex,2ex);
\draw (1ex,1ex) -- (2ex,0);
\draw (2ex,0) -- (2ex,2ex);
\end{tikzpicture}%
}

\newcommand{\PictoFourCycle}{%
\begin{tikzpicture}[baseline=.3ex]
\draw (0,1ex) -- (1ex,2ex);
\draw (1ex,2ex) -- (2ex,1ex);
\draw (2ex,1ex) -- (1ex,0);
\draw (1ex,0) -- (0,1ex);
\end{tikzpicture}%
}

\newcommand{\PictoDiamond}{%
\begin{tikzpicture}[baseline=.3ex]
\draw (0,1ex) -- (1ex,2ex);
\draw (1ex,2ex) -- (2ex,1ex);
\draw (2ex,1ex) -- (1ex,0);
\draw (1ex,0) -- (0,1ex);
\draw (2ex,1ex) -- (0,1ex);
\end{tikzpicture}%
}

\newcommand{\PictoClique}{%
\begin{tikzpicture}[baseline=.3ex]
\draw (0,0ex) -- (2ex,0ex);
\draw (0,0ex) -- (0,2ex);
\draw (2ex,2ex) -- (2ex,0ex);
\draw (2ex,2ex) -- (0,2ex);
\draw (2ex,2ex) -- (0,0);
\draw (2ex,0) -- (0,2ex);
\end{tikzpicture}%
}

\title{Fully Dynamic Four-Vertex Subgraph Counting}

\author{Kathrin Hanauer}{%
University of Vienna, Faculty of Computer Science, Vienna, Austria}{%
kathrin.hanauer@univie.ac.at}{https://orcid.org/0000-0002-5945-837X}{}

\author{Monika Henzinger}{%
University of Vienna, Faculty of Computer Science, Vienna, Austria}{%
monika.henzinger@univie.ac.at}{https://orcid.org/0000-0002-5008-6530}{}

\author{Qi Cheng Hua}{%
University of Vienna, Faculty of Computer Science, Vienna, Austria}{%
}{}{}

\authorrunning{K. Hanauer, M. Henzinger, Q.\,C. Hua} %

\Copyright{Kathrin Hanauer, Monika Henzinger, Qi Cheng Hua} %



\category{} %

\relatedversion{A short version is to appear at SAND`22.} %

\funding{%
\flag{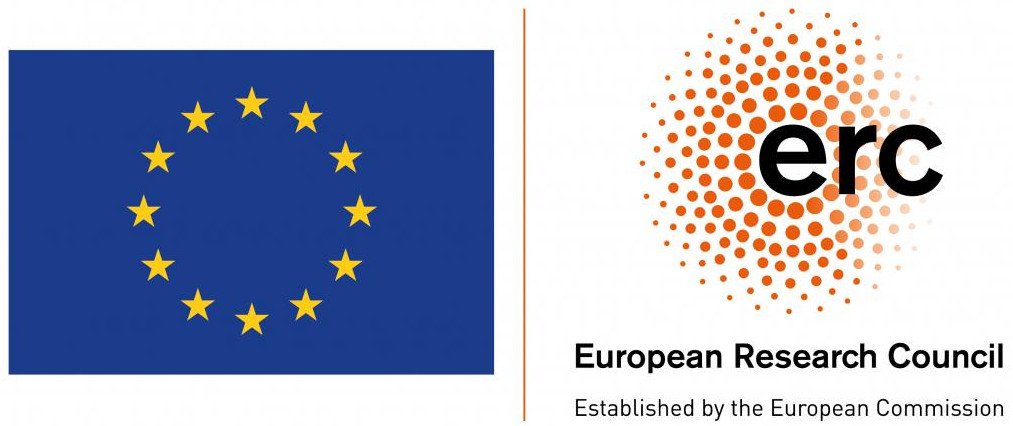}
This project has received funding from the
European Research Council (ERC)
under the European Union's Horizon 2020
research and innovation programme (Grant agreement No.\ 101019564,
``The Design of Modern Fully Dynamic Data Structures (MoDynStruct)''),
as well as from the
Austrian Science Fund (FWF)
and netIDEE SCIENCE project P~33775-N.
}

\acknowledgements{The authors want to thank Leonhard Paul Sidl for careful proofreading.}%

\nolinenumbers %

\begin{document}

\maketitle

\begin{abstract}
This paper presents a comprehensive study of algorithms for maintaining the number of all connected four-vertex
subgraphs in a dynamic graph.
Specifically, our algorithms maintain the number of
paths of length three in deterministic amortized~$\bigO(m^\frac{1}{2})$ update time,
and any other connected four-vertex subgraph which is not a clique
in deterministic amortized update time~$\bigO(m^\frac{2}{3})$.
Queries can be answered in constant time.
We also study the
query times for subgraphs containing an arbitrary edge
that is supplied only with the query
as well as the case where only subgraphs containing a vertex $s$ that is
fixed beforehand are considered.
For length-3 paths, paws, $4$-cycles, and diamonds our bounds match or are
not far from (conditional) lower bounds:
Based on the \OMv{} conjecture we show that any dynamic algorithm that detects
the existence of paws, diamonds, or $4$-cycles
or that counts
length-$3$ paths takes update time $\Omega(m^{1/2-\delta})$.

Additionally, for $4$-cliques and all connected induced subgraphs, we show a
lower bound of $\Omega(m^{1-\delta})$ for any small constant $\delta > 0$ for
the amortized update time,
assuming the static combinatorial $4$-clique conjecture holds.
This shows that the $\bigO(m)$ algorithm by Eppstein
\etal{}~\cite{DBLP:journals/tcs/EppsteinGST12} for these subgraphs cannot be
improved by a polynomial factor.
\end{abstract}

\section{Introduction}\label{sect:introduction}
Detecting or counting subgraphs
is an important question in social network analysis,
where dense subgraphs usually represent communities,
as well as
in telecommunication network surveillance, and computational biology.
This can also be seen in
a recent study by Sahu \etal{}~\cite{Sahu20}: finding and counting fixed subgraphs was the fourth most popular graph computation in practice, only superseded by finding connected components, computing shortest paths, and answering queries about the degree of neighbors.
Furthermore the same study showed that the dynamic setting is important in practice as 65\% of the graphs were dynamic.
Thus, the goal of this paper is to advance the study of subgraph counting problems in dynamic graphs.

Algorithmic problems in dynamic graphs are usually modeled by the following data structure question. Given a potentially non-empty initial graph and a fixed subgraph pattern $\cal P$ (such as a $k$-clique) maintain a data structure that allows the following updates to the current graph $G$: 
\begin{itemize}
\item Insert($u$, $v$): Insert the edge $\Edge{u}{v}$ into $G$.
\item Delete($u$, $v$): Delete the edge $\Edge{u}{v}$ from $G$.
\item Query(): Return the number of subgraphs of pattern $\cal P$ in $G$.
\end{itemize}

Given a subgraph pattern  $\cal P$ that is not a clique, there are two variants of this problem:
One variant, called \emph{induced subgraph counting}, counts a subgraph if it is exactly equivalent to $\cal P$ and does \emph{not} contain any additional edges. In the \emph{non-induced} version, a subgraph is counted if it contains $\cal P$ and potentially additional edges (but \emph{not} additional vertices).
Eppstein and Spiro~\cite{DBLP:journals/jgaa/EppsteinS12} studied subgraph counting for all possible connected three-vertex patterns in both the induced and the non-induced variant and gave a dynamic algorithm with amortized update time $\bigO(h)$, where $h$ is the \emph{$h$-index} of $G$, i.e., the maximum number such that the graph contains $h$ vertices of degree at least $h$. Note that $h$ is $\bigO(\sqrt m)$, where $m$ always is the current number of edges in the graph.
\begin{table}[tb]
\caption{Upper and conditional lower bounds on the time per
update and query for counting different subgraphs,
where $\delta>0$ is an arbitrarily
small constant and $h\in\bigO(\sqrt m)$.
Update times are amortized, query times are worst-case.
Results in \neworbetter{blue} are new or improved.\\
{\footnotesize
${}^*$\,%
Read: For polynomial preprocessing time and $\bigO(\cdot)$ query time,
the update time is $\Omega(\cdot)$.
\\
${}^\dagger$\,%
The previous space complexity for
$3$-cycles and length-$3$ paths was $\bigO(mh)$ with amortized update time $\bigO(h)$~\cite{DBLP:journals/jgaa/EppsteinS12}
and $\bigO(m)$ with amortized update time $\bigO(m^\frac{1}{2})$ for $3$-cycles~\cite{dynamicinsertions};
for (other) $4$-vertex subgraphs, it was $\bigO(mh^2)$ with amortized update time $\bigO(h^2)$~\cite{DBLP:journals/tcs/EppsteinGST12}.
\\
\FromThmCountAny{}\,\Thm{count-any},
\FromThmCountNoninduced{}\,\Thm{count-noninduced},
\FromCorCountTriangles{}\,\Cor{count-triangles},
\FromThmCountNoninducedS{}\,\Thm{count-noninduced-s},
\FromThmCountInduced{}\,\Thm{count-induced},
\FromCorCountInduced{}\,\Cor{count-induced},
\FromThmLowerbounds{}\,\Thm{lowerbounds},
\FromEppsteinTri{}\,\cite{DBLP:journals/jgaa/EppsteinS12},
\FromEppsteinQuad{}\,\cite{DBLP:journals/tcs/EppsteinGST12},
\FromDhulipala{}\,\cite{DLS2020}.
}
}%
\addtolength{\tabcolsep}{-3.5pt}
\label{tab:complexities}
\centering
\begin{tabular}{@{}l@{\;}c@{\;}cc@{\;}cc@{\;}cc@{}}
\toprule
& \multicolumn{2}{c}{\textbf{Lower Bounds}${}^*$}
& \multicolumn{2}{c}{\textbf{Update Time}}
& \multicolumn{2}{c}{\textbf{Query Time}}
& \textbf{Space}${}^\dagger$
\\
\textbf{Subgraph}
& \textbf{Update} & \textbf{Query}
& \textbf{ours} & \textbf{previous}
& \textbf{all} & $e \in E$
& %
\\
\midrule
\multicolumn{8}{@{}l}{\textit{Non-induced subgraphs and $s$-subgraphs}}\\
\midrule
connected, $n=4$
&
&
&
$\bigO(m)$\FromThmCountAny{}
&
&
$\bigO(1)$
&
\neworbetter{$\bigO(1)$}\FromThmCountAny{}
&
\neworbetter{$\bigO(1)$/$\bigO(m)$}\FromThmCountAny{}
\\
claw \PictoClaw
&
$\Omega(1)$
&
$\bigO(1)$
&
$\bigO(1)$
&
$\bigO(1)$\FromEppsteinTri{}\FromEppsteinQuad{}%
&
\notneworbetter{$\bigO(1)$}
&
\notneworbetter{$\bigO(1)$}
&
\notneworbetter{$\bigO(1)$}
\\
length-$3$ path \PictoThreePath
&
\neworbetter{$\Omega(m^{\frac 12-\delta})$}\FromThmLowerbounds{}
&
\neworbetter{$\bigO(m^{1-\delta})$}\FromThmLowerbounds{}
&
\notneworbetter{$\bigO(m^\frac{1}{2})$}\FromThmCountNoninduced{}
&
$\bigO(h)$\FromEppsteinTri{}\FromEppsteinQuad{}%
&
\notneworbetter{$\bigO(1)$}
&
\neworbetter{$\bigO(m^\frac{1}{2})$}\FromThmCountNoninduced{}
&
\notneworbetter{$\bigO(\min(n^2,m^{1.5}))$}\FromThmCountNoninduced{}
\\
paw \PictoPaw
&
\neworbetter{$\Omega(m^{\frac 12-\delta})$}\FromThmLowerbounds{}
&
\neworbetter{$\bigO(m^{1-\delta})$}\FromThmLowerbounds{}
&
\neworbetter{$\bigO(m^\frac 23)$}\FromThmCountNoninduced{}
&
$\bigO(h^2)$\FromEppsteinQuad{}%
&
\notneworbetter{$\bigO(1)$}
&
\neworbetter{$\bigO(m^{\frac{2}{3}})$}\FromThmCountNoninduced{}
&
\neworbetter{$\bigO(n^2)$}\FromThmCountNoninduced{}
\\
$3$-cycle~\PictoThreeCycle
&
\notneworbetter{$\Omega(m^{\frac 12-\delta})$}\FromThmLowerbounds{}
&
\notneworbetter{$\bigO(m^{1-\delta})$}\FromThmLowerbounds{}
&
\notneworbetter{$\bigO(m^\frac{1}{2})$}\FromCorCountTriangles{}
&
$\bigO(h)$\FromEppsteinTri{}%
&
\notneworbetter{$\bigO(1)$}
&
\neworbetter{$\bigO(m^\frac{1}{2})$}\FromCorCountTriangles{}
&
\notneworbetter{$\bigO(\min(n^2,m^{1.5}))$}\FromCorCountTriangles{}
\\
$4$-cycle \PictoFourCycle
&
\neworbetter{$\Omega(m^{\frac 12-\delta})$}\FromThmLowerbounds{}
&
\neworbetter{$\bigO(m^{1-\delta})$}\FromThmLowerbounds{}
&
\neworbetter{$\bigO(m^\frac 23)$}\FromThmCountNoninduced{}
&
$\bigO(h^2)$\FromEppsteinQuad{}%
&
\notneworbetter{$\bigO(1)$}
&
\neworbetter{$\bigO(m^{\frac{2}{3}})$}\FromThmCountNoninduced{}
&
\neworbetter{$\bigO(n^2)$}\FromThmCountNoninduced{}
\\
$k$-cycle, $k\geq 5$
&
\neworbetter{$\Omega(m^{\frac 12-\delta})$}\FromThmLowerbounds{}
&
\neworbetter{$\bigO(m^{1-\delta})$}\FromThmLowerbounds{}
\\
diamond \PictoDiamond
&
\neworbetter{$\Omega(m^{\frac 12-\delta})$}\FromThmLowerbounds{}
&
\neworbetter{$\bigO(m^{1-\delta})$}\FromThmLowerbounds{}
&
\neworbetter{$\bigO(m^\frac 23)$}\FromThmCountNoninduced{}
&
$\bigO(h^2)$\FromEppsteinQuad{}%
&
\notneworbetter{$\bigO(1)$}
&
\neworbetter{$\bigO(m^{\frac{2}{3}})$}\FromThmCountNoninduced{}
&
\neworbetter{$\bigO(\min(nm,m^\frac{5}{3}))$}\FromThmCountNoninduced{}
\\
$4$-clique~\PictoClique
&
\neworbetter{$\Omega(m^{1-\delta})$}\FromThmCountInduced{}
&
\neworbetter{$\bigO(m^{2-\delta})$}\FromThmCountInduced{}
&
\notneworbetter{$\bigO(m)$}\FromDhulipala{}
&
$\bigO(h^2)$\FromEppsteinQuad{}%
&
\notneworbetter{$\bigO(1)$}
&
\notneworbetter{$\bigO(m)$}\FromDhulipala{}
&
\notneworbetter{$\bigO(1)$}
\\
\midrule
$s$-claw \PictoClaw
&
$\Omega(1)$
&
$\bigO(1)$
&
\neworbetter{$\bigO(1)$}\FromThmCountNoninducedS{}
&
&
\neworbetter{$\bigO(1)$}\FromThmCountNoninducedS{}
&
&
\neworbetter{$\bigO(1)$}\FromThmCountNoninducedS{}
\\
$s$-length-$3$-path \PictoThreePath
&
\neworbetter{$\Omega(m^{\frac 12-\delta})$}\FromThmLowerbounds{}
&
\neworbetter{$\bigO(m^{1-\delta})$}\FromThmLowerbounds{}
&
\neworbetter{$\bigO(m^{\frac{1}{2}})$}\FromThmCountNoninducedS{}
&
&
\neworbetter{$\bigO(1)$}\FromThmCountNoninducedS{}
&
&
\neworbetter{$\bigO(n)$}\FromThmCountNoninducedS{}
\\
$s$-paw \PictoPaw
&
\neworbetter{$\Omega(m^{\frac 12-\delta})$}\FromThmLowerbounds{}
&
\neworbetter{$\bigO(m^{1-\delta})$}\FromThmLowerbounds{}
&
\neworbetter{$\bigO(m^\frac 23)$}\FromThmCountNoninducedS{}
&
&
\neworbetter{$\bigO(1)$}\FromThmCountNoninducedS{}
&
&
\neworbetter{$\bigO(n^2)$}\FromThmCountNoninducedS{}
\\
$s$-$3$-cycle~\PictoThreeCycle
&
\notneworbetter{$\Omega(m^{\frac 12-\delta})$}\FromThmLowerbounds{}
&
\notneworbetter{$\bigO(m^{1-\delta})$}\FromThmLowerbounds{}
&
\neworbetter{$\bigO(m^{\frac 12})$}\FromThmCountNoninducedS{}
&
&
\neworbetter{$\bigO(1)$}\FromThmCountNoninducedS{}
&
&
\neworbetter{$\bigO(n)$}\FromThmCountNoninducedS{}
\\
$s$-$4$-cycle~\PictoFourCycle
&
&
&
\neworbetter{$\bigO(m^{\frac 23})$}\FromThmCountNoninducedS{}
&
&
\neworbetter{$\bigO(1)$}\FromThmCountNoninducedS{}
&
&
\neworbetter{$\bigO(n^2)$}\FromThmCountNoninducedS{}
\\
$s$-$k$-cycle, $k\geq 5$, odd
&
\neworbetter{$\Omega(m^{\frac 12-\delta})$}\FromThmLowerbounds{}
&
\neworbetter{$\bigO(m^{1-\delta})$}\FromThmLowerbounds{}
&
&
\\
$s$-diamond \PictoDiamond
&
\neworbetter{$\Omega(m^{\frac 12-\delta})$}\FromThmLowerbounds{}
&
\neworbetter{$\bigO(m^{1-\delta})$}\FromThmLowerbounds{}
&
\neworbetter{$\bigO(m^\frac 23)$}\FromThmCountNoninducedS{}
&
&
\neworbetter{$\bigO(1)$}\FromThmCountNoninducedS{}
&
&
\neworbetter{$\bigO(n^2)$}\FromThmCountNoninducedS{}
\\
$s$-$4$-clique \PictoClique{}
&
\neworbetter{$\Omega(m^{\frac 12-\delta})$}\FromThmLowerbounds{}
&
\neworbetter{$\bigO(m^{1-\delta})$}\FromThmLowerbounds{}
&
\neworbetter{$\bigO(m)$}\FromThmCountNoninducedS{}
&
&
\neworbetter{$\bigO(1)$}\FromThmCountNoninducedS{}
&
&
\neworbetter{$\bigO(1)$}\FromThmCountNoninducedS{}
\\
\midrule
\multicolumn{8}{@{}l}{\textit{Induced subgraphs}}\\
\midrule
connected, $n=4$
&
\neworbetter{$\Omega(m^{1-\delta})$}\FromThmCountInduced
&
\neworbetter{$\bigO(m^{2-\delta})$}\FromThmCountInduced
&
\notneworbetter{$\bigO(m)$}\FromCorCountInduced{}
&
$\bigO(h^2)$\FromEppsteinQuad{}%
&
\notneworbetter{$\bigO(1)$}
&
&
\neworbetter{$\bigO(1)$}\FromCorCountInduced{}
\\
\bottomrule
\end{tabular}
\end{table}

There are six connected graphs on four vertices, which we refer to as
\emph{length-$3$ path}~\PictoThreePath, \emph{claw}~\PictoClaw, \emph{paw}~\PictoPaw,
\emph{$4$-cycle}~\PictoFourCycle, \emph{diamond}~\PictoDiamond, and
\emph{$4$-clique}~\PictoClique{}.
Eppstein et al.~\cite{DBLP:journals/tcs/EppsteinGST12} extended the method  of~\cite{DBLP:journals/jgaa/EppsteinS12} to maintain counts of any (induced and non-induced) connected
four-vertex subgraph in amortized time $\bigO(h^2) = \bigO(m)$ and $\bigO(mh^2)$ space.
This paper contains a comprehensive study of the complexity of dynamically
counting all possible connected  four-vertex subgraphs.
We present new improved dynamic algorithms and give the first conditional lower
bounds.

\textbf{Upper bounds.}
We show how to maintain the number of any connected four-vertex non-induced
subgraph that is not a clique (such as a paw, a $4$-cycle, or a diamond) in
update time $\bigO(m^{2/3})$ and at most $\bigO(nm)$ space.
For graphs with an $h$-index
larger than $\bigO(m^{1/3})$,
our algorithms are hence faster than the
$\bigO(h^2)$ algorithm by Eppstein
\etal{}~\cite{DBLP:journals/tcs/EppsteinGST12}.
Besides,
our data structure can also be used to count all $s$-triangles, \ie, triangles
that contain a fixed vertex $s$, in $\bigO(m^{1/2})$ update time, constant
query time, and $\bigO(n)$ space, and likewise for $s$-length-$3$ paths.
The update time is in $\bigO(m^{2/3})$ for $s$-$4$-cycles, $s$-paws, and
$s$-diamonds, with $\bigO(n^2)$ space,
and $\bigO(m)$ for $s$-$4$-cliques with constant space.
For $\PartFrac \in \Range{0}{1}$, our data structure supports queries on the
number of triangles containing an arbitrary vertex or edge in
$\bigO(\min(m^{2\PartFrac},n^2))$ or $\bigO(\min(m^{1-\PartFrac},n))$
worst-case time, respectively, with an amortized update time of
$\bigO(m^{\max(\PartFrac,1-\PartFrac)})$ or $\bigO(m^{\PartFrac})$,
respectively, and $\bigO(n^2)$ space.
We also show how to maintain length-$3$ paths in time $\bigO(m^{1/2})$, but
this result was already stated in~\cite{DBLP:journals/tcs/EppsteinGST12}.
See \Table{complexities} for an overview.
All our algorithms are deterministic and the running time bounds are amortized
unless stated otherwise.

\textbf{Lower bounds.}
We also give the first conditional lower bounds for counting various
four-vertex subgraphs based on two popular hypotheses: the Online Boolean
Matrix-Vector Multiplication (\OMv{}) conjecture~\cite{henzinger2015unifying}
and the Combinatorial $k$-Clique hypothesis.
In the \OMv{} conjecture we are given a Boolean $n\times n$ matrix $M$ that can be
preprocessed.
Then, an online sequence of vectors $v_1,\dots ,v_n$ is presented and the goal
is to compute each Boolean product $M v_i$ (using conjunctions and
disjunctions) before seeing the next vector $v_{i+1}$.

\begin{conjecture}[OMv]\ConLabel{omv}
For any constant $\delta > 0$, there is no $\bigO(n^{3-\delta})$-time
algorithm that solves \OMv{} with error probability at most $1/3$ in the word-RAM model with $O(\log n)$ bit words.
\end{conjecture}
Based on the \OMv{} conjecture we show that \emph{detecting} (with probability at least 2/3 in the word-RAM model with $O(\log n)$ bit words) the existence of (non-induced) paws, diamonds, 4-cliques, or $k$-cycles for any $k \ge 3$  in a graph with edge insertions and deletions takes amortized update time $\Omega(m^{1/2-\delta})$ or  query time
$\Omega(m^{1-\delta})$ if only polynomial preprocessing time is allowed. This lower bound applies also to the worst-case update time of any insertions-only or deletions-only algorithm. Note that this lower bound does not only apply to \emph{counting} the number of such subgraphs but already to \emph{detecting} whether such a subgraph exists. Let $s$ be a fixed vertex in the graph. The same lower bounds apply to algorithms that detect whether a diamond, $4$-clique, or $k$-cycle with odd $k$ containing $s$ exists. Finally, we also show a lower bound for \emph{counting} the number of length-3 paths and length-3 $s$-paths.
We remark that the conditional lower bounds for ($s$-)$3$-cycles were already known before~\cite{henzinger2015unifying}.

We also use the Combinatorial $k$-Clique hypothesis which is defined as follows and has become popular in recent years (e.g.~\cite{lincoln2018tight,abboud2018if,bringmann2019fine,DBLP:conf/soda/BergamaschiHGWW21}).

\begin{conjecture}[Combinatorial $k$-Clique]\ConLabel{k-clique}
For any constant $\delta >0$, for an $n$-vertex graph there is no $\bigO(n^{k-\delta})$ time combinatorial algorithm for $k$-clique detection with error probability at most $1/3$ in the word-RAM model with $\bigO(\log n)$ bit words.
\end{conjecture}

Let $\delta > 0$ be a small constant.
Based on the 4-clique conjecture we show that (with probability at least 2/3 in the word-RAM model with $O(\log n)$ bit words) there does not exist a combinatorial algorithm that counts \emph{any connected induced four-vertex subgraph} in a dynamic graph with amortized update time $O(n^{4-2\delta}/m)$, which is $O(m^{1-\delta})$, and query preprocessing time $O(n^{4-2\delta})$. 
This bound applies also to any insertions-only algorithm.
The bound can be extended to any $k$-clique with $k >4$ showing that the amortized update time is 
$\Omega(n^{k-2\delta}/m)$
with $\Omega(n^{k-2\delta})$ preprocessing and query time.

\textbf{Technical contribution.}
For the upper bounds we extend and improve upon Eppstein
\etal{}~\cite{DBLP:journals/tcs/EppsteinGST12} both with respect to running
time and space.
The high-level idea is as follows:
We partition the vertices into (few) \emph{high-degree} and (many)
\emph{low-degree} vertices and then maintain for each vertex, vertex pair, or
vertex triple certain information in a data structure such as the number of
certain paths up to length 3 that contain low-degree vertices.
When an edge $\{u,v\}$ is updated, four-vertex subgraphs that contain $u$, $v$,
and two other low-degree vertices can be quickly counted using the information
in the data structure.
On the other side, subgraphs that contain a high-degree vertex in addition
to $u$ and $v$ can often be counted ``from scratch'' after each update as there
are few high-degree vertices.
The more challenging case is the situation where relationships involving two or
more high-degree vertices in the subgraph need to be checked or maintained.
How to deal with this depends on the subgraph to count.
For diamonds, e.g., this requires to keep certain information about triples of
vertices.

For the conditional lower bounds based on the combinatorial $k$-clique
conjecture we first directly deduce the lower bound for incremental 4-clique
counting.
Then we use the fact that (a)~we have a lower bound for 4-cliques, (b)~we
developed algorithms with $\bigO(m^{2/3})$ update time for all non-induced
subgraphs, and (c)~there exist ``counting formulas'' that allow to compute the
number of any induced subgraph based on the number of 4-cliques and the number
of non-induced subgraphs.
Thus, if the number of an induced subgraph pattern could be computed in
$\bigO(m^{1-\delta})$ time per update for some small $\delta > 0$ then we could
use the corresponding counting formula and our algorithms for non-induced
subgraphs to dynamically maintain the number of 4-cliques, contradicting our
dynamic lower bound for 4-cliques.

For the conditional lower bounds based on  the \OMv{} conjecture we construct
for each subgraph pattern $\cal P$ based on an \oneuMv{} instance (which is a
variant of \OMv{}) a suitable graph based on $\cal P$  with $\bigO(n)$ vertices
and $\bigO(n^2)$ edges such that detecting the existence of the (non-induced)
version of $\cal P$ in the graph equals finding the answer for the \oneuMv{}
instance.
Then the lower bound follows as in~\cite{henzinger2015unifying}.
The challenge is to construct such a graph. We show how to do this for
detecting non-induced ($s$)-paws, ($s$)-diamonds, ($s$)-4-cliques, and
($s$)-$k$-cycles for $k\ge 3$ and for counting non-induced length-3
($s$)-paths.

Our paper gives in Section~\ref{sect:preliminaries}  the preliminaries, in
Section~\ref{sect:algorithms} our new algorithms, and in
Section~\ref{sect:lowerbounds} our lower bounds.
Proofs that have been omitted in the main part are given in the appendix.
\section{Preliminaries}\label{sect:preliminaries}
\paragraph*{Basic Definitions}\SectLabel{basics}
We consider an undirected dynamic graph $G = (V, E)$ and use $n$ to denote the number of vertices and $m$ for the current number of edges.
Two vertices $u \neq v$ are \emph{adjacent} if there is an edge $e =
\Edge{u}{v} \in E$.
In this case, $u$ and $v$ are \emph{incident} to $e$.
The \emph{neighborhood} $\Neigh(v)$ of a vertex $v$ is defined as $\Set{ u \mid
\Edge{u}{v} \in E}$ and $v$'s \emph{degree} is $\Degree(v) = \Card{\Neigh(v)}$.
As a shorthand notation to exclude just one vertex, we use
$\NeighEx{w}(v) = \Neigh(v) \setminus \Set{w}$ and $\DegreeEx{w}(v) =
\Card{\NeighEx{w}(v)}$, \ie, $\DegreeEx{w}(v) = \Degree(v) - 1$ if $w \in
\Neigh(v)$
and $\DegreeEx{w}(v) = \Degree(v)$ otherwise.
A $k$-\emph{path} (also length-$k$ path) is a sequence of distinct edges
$\Sequence{\Edge{v_0}{v_1}, \Edge{v_1}{v_2}, \dots, \Edge{v_{k-1}}{v_k}}$ of
length $k$, %
where $v_i \neq v_j$ for all $0 \leq i,j \leq k$.
A $k$-\emph{cycle} is a $k$-path where as an only exception the first vertex
equals the last, \ie, $v_0 = v_k$.
A $3$-cycle is also called \emph{triangle}.
A \emph{claw} is a graph consisting of a vertex $x$, called the \emph{central}
vertex, and three edges incident to it.
A \emph{paw} is a graph consisting of a triangle together with an additional edge attached to one of the vertices of the triangle.
This vertex is called the \emph{central} vertex of the paw and the additional edge the \emph{arm}.
A \emph{diamond} is a $4$-cycle with a \emph{chord}, \ie, an additional edge
connecting one of the two pairs of non-adjacent vertices, which creates two
triangles sharing the chordal edge.
A $k$-\emph{clique} is the complete graph $K_k$ on $k$ vertices.

A graph $G' = (V', E')$ is a \emph{subgraph} of $G$ if $V' \subseteq V$ and $E'
\subseteq E$.
A subgraph $G' = (V', E')$ is said to be \emph{induced} if $E' = \Set{
\Edge{u}{v} \in E \mid u, v \in V' }$.
The term \emph{non-induced subgraph} or just \emph{subgraph} without an
adjective refers to all subgraphs, induced and not.
For a static graph $\Graphlet$, also called \emph{pattern},  we denote by
$\Count(G, \Graphlet)$ the number of (non-induced) subgraphs of $G$ that are
isomorphic to $\Graphlet$, and by $\ICount(G, \Graphlet)$ the number of induced
subgraphs that are isomorphic to $\Graphlet$, in each case divided by the
number of automorphisms of $\Graphlet$.
For a vertex $s \in V$, we further denote by $\Count(G, \Graphlet, s)$
the number of non-induced subgraphs of $G$ that are isomorphic to $\Graphlet$
and contain $s$.
We denote by $\GraphletSet_k$ the set of connected subgraphs on $k$ vertices.
There are six connected graphs in $\GraphletSet_4$, which we refer to as
\emph{length-$3$ path}~\PictoThreePath, \emph{claw}~\PictoClaw, \emph{paw}~\PictoPaw,
\emph{$4$-cycle}~\PictoFourCycle, \emph{diamond}~\PictoDiamond, and
\emph{$4$-clique}~\PictoClique{}.

Given a pattern $\Graphlet$, we study the problem of \emph{maintaining} the
number of occurrences of $\Graphlet$ as a subgraph or induced subgraph of $G$,
called the \emph{non-induced subgraph count} $\Count(G, \Graphlet)$ or the
\emph{induced subgraph count} $\ICount(G, \Graphlet)$, respectively, of
$\Graphlet$ in $G$,
and the analogous problem of \emph{maintaining} the count of non-induced subgraphs
containing a specific predefined vertex $v \in V$ or edge $e \in E$,
$\Count(G, \Graphlet, v)$
or $\Count(G, \Graphlet, e)$, respectively.
Unless stated otherwise, \emph{maintaining} a count implies that we can
retrieve it by a query in constant time.
We also study the closely related problem of \emph{querying} the number
of subgraphs containing a specific vertex or edge that is given
only with the query.
\paragraph*{Further Related Work}\label{sect:related}
Detecting (or counting) subgraphs, also known as the \textit{subgraph
isomorphism} problem, generalizes the \textit{clique} or \textit{Hamiltonian
cycle} problem and is hence $\NP$-hard.
Nevertheless, it can be solved efficiently if the subgraphs to detect or count
are restricted.
\subparagraph*{Static algorithms.}
Algorithms counting numbers of subgraphs and solving related problems have been
studied extensively for static graphs.
Alon, Yuster and Zwick~\cite{AYZ} developed an algorithm to count the number of
triangles and other circles (up to seven vertices) in a graph in time
$\bigO(n^\omega)$, where $n$ denotes the number of vertices and
$\omega<2.373$~\cite{LeGall14} is the fast matrix multiplication exponent, \ie,
the smallest value such that two $n\times n$ matrices can be multiplied in
$\bigO(n^\omega)$ time.
Kloks, Kratsch and M\"uller~\cite{Kloks} showed how to compute the number of
4-cliques in time $\bigO(m^{(\omega+1)/2})$ ($m$ denotes the number of edges)
and the number of any other subgraph of size 4 in time $\bigO(n^\omega+m^{(\omega+1)/2})$.
There is also a large body of work on  parallel algorithms for
counting subgraphs (see \eg~\cite{mapreduce}) and to develop approximate
algorithms in the streaming setting (see \eg~\cite{streaming1,streaming2}).
\subparagraph*{Dynamic algorithms.} A more recent development is counting subgraph numbers for dynamic graphs.
Kara \etal{}~\cite{dynamicinsertions} provided an algorithm for counting
triangles in amortized time $\bigO(\sqrt{m})$ per update and $\bigO(m)$
space, which can also enumerate them with constant time delay.
Dhulipala \etal{}~\cite{DLS2020} extended it to a batch-dynamic
parallel algorithm with $\bigO(\Delta\sqrt{\Delta+m})$ amortized work and
$\bigO(\polylog(\Delta+m))$ depth \whp{} for a batch of $\Delta$ updates.
Based on a static algorithm to enumerate cliques,
they show how to obtain a dynamic algorithm
for maintaining the number of $k$-cliques for a fixed $k > 3$
with expected $\bigO(\Delta(m+\Delta)\alpha^{k-4})$ work and
$\bigO(\log^{-2}n)$ depth \whp{} per update and $\bigO(m+\Delta)$ space, where
$\alpha\in\bigO(\sqrt{m})$ is the arboricity of the graph.
They also give a parallel fast matrix multiplication algorithm with
$\bigO(\min(\Delta m^{(2k-1)\omega_p/(3\omega_p+3)}, {(\Delta+m)}^{2(k+1)\omega_p/(3\omega_p+3)}))$
amortized work and $\bigO(\log(\Delta+m))$ depth,
with parallel matrix multiplication constant $\omega_p$.
Eppstein et al.~\cite{DBLP:journals/tcs/EppsteinGST12} also count all
three-vertex subgraphs in directed graphs in amortized time $\bigO(h)$.
For specific graph classes, namely \emph{bounded expansion} graphs, Dvorak and Tuma~\cite{DvorakDynamicSubgraph} gave a different algorithm for
maintaining counts for arbitrary graph patterns $\cal P$  of $k$ vertices the number of induced subgraphs of pattern $\cal P$ in amortized time
$\bigO(\log^{(k^2-k)/2 -1}n)$ and in amortized time
$\bigO(n^{\epsilon})$ for any constant $\epsilon > 0$ in
\emph{no-where dense} graphs.

In recent subsequent work~\cite{HLS2022}, it was shown that counting $4$-cycles
is hard also in random graphs, \ie, with $\Omega(m^{1/2-\delta})$ update
time or $\Omega(m^{1-\delta})$ query time.
\section{New Counting Algorithms for Subgraphs on Four
Vertices}\label{sect:algorithms}
Counting the number of claws~\cite{DBLP:journals/jgaa/EppsteinS12} and
$4$-cliques is fairly straightforward~\cite{DLS2020}.
The latter extends to $4$-vertex subgraphs in general, where all work is
either done during the updates or queries.
\begin{observation}[\cite{DBLP:journals/jgaa/EppsteinS12}]
Let $G = (V, E)$ be a dynamic graph and
$\Graphlet$ be the claw~\PictoClaw{}.
Then, $\Count(G, \Graphlet) =\sum_{v \in V, \Degree(v) \ge 3} \binom{\Degree(v)}{3}$.
The count can be initialized in $\bigO(n)$ time and maintained
in constant time and space.
We can query $\Count(G, \Graphlet, \Edge{u}{v})$
for an arbitrary edge $\Edge{u}{v} \in E$ in constant time.
\end{observation}
\begin{theorem}\ThmLabel{count-any}
Let $G = (V, E)$ be a dynamic graph and
$\Graphlet \in \GraphletSet_4$.
We can
\begin{enumerate}[(i)]
\item
maintain $\Count(G, \Graphlet)$ in worst-case $\bigO(m)$ update time and
constant space.
\item
query $\Count(G, \Graphlet, \Edge{u}{v})$
for an arbitrary edge $\Edge{u}{v} \in E$
in worst-case $\bigO(m)$ time, with
constant update time and space.
\item
query $\Count(G, \Graphlet, \Edge{u}{v})$
for an arbitrary edge $\Edge{u}{v} \in E$
in worst-case constant time,
with $\bigO(m)$ worst-case update time and $\bigO(m)$ space.
\end{enumerate}
\end{theorem}

For connected $4$-vertex subgraphs other than the claw and the clique, we
can invest space to achieve speedups in running time.
In the following, we present our data structure and show how to update it
efficiently.
We then use different parts of the data structure to count different subgraphs.
Specifically, we prove the following result:
\begin{theorem}\ThmLabel{count-noninduced}
Let $G$ be a dynamic graph and $\Graphlet \in \GraphletSet_4$.
We can maintain
$\Count(G, \Graphlet)$ %
in
\begin{enumerate}[(i)]
\item
amortized $\bigO(\sqrt{m})$ update time
with $\bigO(\min(m^{1.5}, n^2))$ space
and query $\Count(G, \Graphlet, e)$ for an arbitrary edge $e \in E$
in worst-case $\bigO(\sqrt{m})$ time
if $\Graphlet$ is the length-$3$ path~\PictoThreePath{}, %
\item
amortized $\bigO(m^{2/3})$ update time
with $\bigO(n^2)$ space
and query $\Count(G, \Graphlet, e)$ for an arbitrary edge $e \in E$
in worst-case $\bigO(m^{2/3})$ time
if $\Graphlet$ is the paw~\PictoPaw{}
or the $4$-cycle~\PictoFourCycle{}, %
\item
amortized $\bigO(m^{2/3})$ update time
with $\bigO(\min(nm,m^{5/3}))$ space
and query $\Count(G, \Graphlet, e)$ for an arbitrary edge $e \in E$
in worst-case $\bigO(m^{2/3})$ time
if $\Graphlet$ is the diamond~\PictoDiamond{}.
\end{enumerate}
\end{theorem}

Our algorithm makes use of a standard technique in dynamic graph algorithms
that partitions vertices into \emph{high-degree} and \emph{low-degree}
vertices.
We adapt it to our needs as follows:
Let $m_0$ be the number of edges of $G$ at construction or when recomputing
from scratch and let $M = 2m_0$. A recomputation from scratch and re-initialization of the partition is triggered whenever the current number of edges $m < \lfloor\frac{M}{4}\rfloor$ or $m \geq M$.
Initially and at each recomputation from scratch, a vertex $v \in V$ is
classified as \emph{high-degree} and added to partition $\High$ if $\deg(v)
\geq \theta$ and otherwise as \emph{low-degree} and added to partition $\Low :=
V \setminus \High$, for some threshold $\theta$.
As $G$ evolves, a high-degree vertex $v$ is reclassified as low and moved to
$\Low$ only if $\Degree(v) < \frac{1}{2}\theta$.
Vice-versa, a low-degree vertex $v$ is reclassified as high and moved to $\High$
only if $\Degree(v) \geq \frac{3}{2}\theta$.
We call such a partition
$(\High, \Low)$ a \emph{dynamic vertex partition with
threshold $\theta$}.
If $\theta = M^\PartFrac$ for some $\PartFrac \in \Range{0}{1}$, we call the
partition an \emph{$\PartFrac$-partition.}
\begin{theorem}[$\PartFrac$-Partition~\cite{DBLP:journals/tods/KaraNNOZ20}]%
\ThmLabel{partition}
Let $\PartFrac \in \Range{0}{1}$ and consider  an $\PartFrac$-partition $(\High, \Low)$
for a dynamic graph $G = (V, E)$.
Then, $\Card{\High} \in \bigO(m^{1-\PartFrac})$.
The partition can be constructed in $\bigO(n)$ time and maintained in amortized
constant time per update with amortized $\bigO(m^{-\PartFrac})$ changes to the
partition per update and $\Omega(m)$ updates between two recomputations from scratch.
The required space is $\bigO(n)$.
\end{theorem}

\begin{figure}[tb]
\centering
\includegraphics[width=\textwidth]{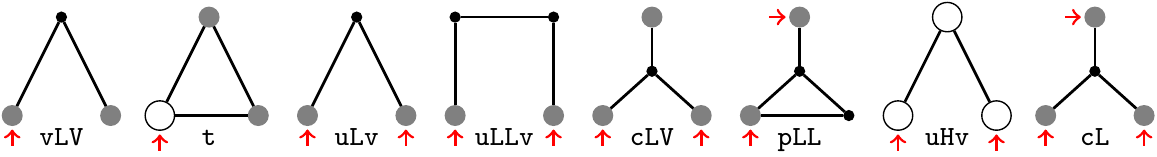}
\caption{Subgraph structures of $\DS_{\PartFrac}$. Small and filled vertices have low degree, large and empty vertices high degree,
medium-sized and shaded vertices can have either high or low degree. Anchors are marked by arrows.}%
\label{fig:substructures}
\end{figure}

\paragraph*{Data Structure $\DS_{\PartFrac}$}
We assume that the algorithm can access the degree of a vertex $v$ in constant
time and, for each pair of vertices $u, v$ determine in constant time whether
$\Edge{u}{v} \in E$.
In addition, we maintain the following data structure $\DS_{\PartFrac}$ or a subset of it,
if we are not only interested in counting some specific subgraphs on four vertices.
All subgraph structures that are part of $\DS_{\PartFrac}$ are non-induced.
See \Figure{substructures} for visualizations.
\begin{itemize}
\item an $\PartFrac$-partition $(\High, \Low)$ %
\item For each vertex $v \in V$: %
\Aux{vLV[$v$]}:
the number of $2$-paths $\Sequence{\Edge{v}{x}, \Edge{x}{y}}$ with $x \in \Low$
\item For each vertex $v \in \High$: %
\Aux{t[$v$]}:
the number of $3$-cycles $\Sequence{\Edge{v}{x}, \Edge{x}{y}, \Edge{y}{v}}$
\item For each distinct, unordered pair of vertices $u, v \in V$: %
\begin{itemize}
\item
\Aux{uLv[$u, v$]}:
the number of $2$-paths $\Sequence{\Edge{u}{x}, \Edge{x}{v}}$ with $x \in \Low$
\item %
\Aux{uLLv[$u, v$]}:
the number of length-$3$ paths $\Sequence{\Edge{u}{x}, \Edge{x}{y}, \Edge{y}{v}}$ with $x,y \in \Low$
\item %
\Aux{cLV[$u, v$]}:
the number of claws with a central vertex $x \in \Low$, $u \neq  x \neq v$
\item %
\Aux{pLL[$u, v$]}:
the number of paws with a central vertex $x \in \Low$, $u \neq  x \neq v$,
$u$ or $v$ at the other end of the arm, and fourth vertex $y \in \Low$
\end{itemize}
\item For each distinct, unordered pair of vertices $u, v \in \High$: %
\Aux{uHv[$u, v$]}:
the number of length-$2$ paths $\Sequence{\Edge{v}{x}, \Edge{x}{v}}$ with $x \in \High$
\item For each distinct, unordered triple of vertices $u, v, w \in V$: %
\Aux{cL[$u, v, w$]}:
the number of claws with a fourth vertex $x \in \Low$ at the center
\end{itemize}
For each auxiliary subgraph whose count is maintained by the data structure,
we call vertices of the set that acts as key \emph{anchors}, \eg,
$u$ and $v$ are anchors for the $2$-paths $\Sequence{\Edge{u}{x}, \Edge{x}{v}}$ with $x \in \Low$,
which are counted by \Aux{uLv[$u, v$]}.
We use hash tables with $\bigO(1)$ amortized access time
and only store non-zero counts.
Note that
\Aux{uLv}, \Aux{uLLv}, \Aux{pLL}, and \Aux{cL}
correspond to $s_2$, $s_3$, $s_5$, and $s_7$, respectively,
in the algorithm by Eppstein \etal{}~\cite{DBLP:journals/tcs/EppsteinGST12},
whereas
\Aux{vLV}, \Aux{t}, and \Aux{cLV}
are modifications
of $s_0$, $s_1$, and $s_4$, and
\Aux{uHv}
has no equivalent at all.
However, Eppstein \etal{}~\cite{DBLP:journals/tcs/EppsteinGST12}
use a different partitioning scheme, where there are at most $\bigO(h)$ vertices of
degree $\Omega(h)$, whereas in our case,
there are at most $\bigO(m^{1-\PartFrac})$ vertices of degree
$\Omega(m^{\PartFrac})$, which requires a different running time analysis also
for the common auxiliary counts.

We generally assume that in case of an edge insertion, the auxiliary counts are
updated immediately \emph{before} the counts of interest, and in reverse order
for an edge deletion.
The update of the $\PartFrac$-partition can either happen first or last (but
not in between).
We also assume that we start with an empty graph and all counts are initialized
to zero.
\paragraph*{Maintaining the Data Structure $\DS_{\PartFrac}$}
Given a dynamic graph $G = (V,E)$ and $\PartFrac \in \Range{0}{1}$, we show how
the components of the data structure $\DS_{\PartFrac}$ can be updated after an
edge insertion or deletion and if a vertex changes partition.
We start with a helper lemma:
\begin{lemma}\LemLabel{maintain-aux}
Let \Aux{aux} be an auxiliary subgraph count in $\DS_{\PartFrac}$ with
worst-case update time $\mathcal{E}_\Aux{aux}$ after an edge insertion or
deletion, worst-case update time $\mathcal{V}_\Aux{aux}$ after a vertex
changes partition, and $\mathcal{S}_\Aux{aux}$ space.
Then, $\DS_{\PartFrac}$ with \Aux{aux} can be maintained in amortized update time
$\bigO(\mathcal{E}_\Aux{aux} + \mathcal{V}_\Aux{aux}\cdot m^{-\PartFrac})$
with $\bigO(n+\mathcal{S}_\Aux{aux})$ space.
\end{lemma}
\begin{proof}
By \Theorem{partition}, the $\PartFrac$-partition can be maintained in
$\bigO(n)$ space and such that there are $\Omega(m)$ updates between two
recomputations of the partition from scratch.
After each such complete repartitioning, we set \Aux{aux} to zero and re-insert
all edges one-by-one.
The total recomputation time hence is $\bigO(m \cdot \mathcal{E}_\Aux{aux})$
and amortization over $\Omega(m)$ edge updates results in an amortized edge
update time of $\bigO(\mathcal{E}_\Aux{aux})$.
By \Theorem{partition}, there are amortized $\bigO(m^{-\PartFrac})$ vertices
changing partition per edge update, hence the claim follows.
\end{proof}

As the insertion and deletion operations are entirely symmetric and only differ
in whether a certain amount is added or subtracted from the stored counts, we
only give the details for edge insertions in the following.
Similarly, we only consider the case that a vertex $v$ changes from $\Low$ to
$\High$; the other case is symmetric.
Note that if $v$ is about to change partitions, $\Degree(v) \in
\Theta(m^\PartFrac)$.

\begin{lemma}\LemLabel{maintain-vLV}
$\DS_{\PartFrac}$ with \Aux{vLV} can be maintained in amortized
$\bigO(m^{\PartFrac})$ update time and $\bigO(n)$ space.
\end{lemma}
\begin{proof}
Let $\Edge{u}{v}$ be the newly inserted edge.
If $u \in \Low$ ($v \in \Low$), increase \Aux{vLV[$w$]} by one for each $w \in
\NeighEx{v}(u)$ ($w \in \NeighEx{u}(v)$) and increase \Aux{vLV[$v$]} by
$\Degree(u) - 1$ (\Aux{vLV[$u$]} by $\Degree(v) - 1$).
This takes $\bigO(m^{\PartFrac})$ time.

If a vertex $v \in \Low$ changes to $\High$, this affects all length-$2$ paths
where $v$ is the central, low-degree vertex.
For each neighbor $w \in \Neigh(v)$, decrease \Aux{vLV[$w$]} by $\Degree(v)-1$.
The running time is $\bigO(\Degree(v)) = \bigO(m^\PartFrac)$.

As each vertex may be adjacent to at least one low-degree vertex, the space
requirement is $\bigO(n)$.
By \Lemma{maintain-aux}, $\DS_{\PartFrac}$ with \Aux{vLV} can hence be
maintained in $\bigO(m^{\PartFrac})$ amortized update time and $\bigO(n)$ space.
\end{proof}
\begin{lemma}\LemLabel{maintain-uLv}
$\DS_{\PartFrac}$ with \Aux{uLv} can be maintained in amortized
$\bigO(m^{\PartFrac})$ time per update and $\bigO(\min(m^{1+\PartFrac},n^2))$ space.
\end{lemma}
\begin{proof}
Let $\Edge{u}{v}$ be the newly inserted edge.
If $u \in \Low$ ($v \in \Low$),
increment \Aux{uLv[$v, w$]} (\Aux{uLv[$u, w$]}) by
one for each $w \in \NeighEx{v}(u)$ ($w \in \NeighEx{u}(v)$).
This takes $\bigO(m^{\PartFrac})$ time.

If a vertex $v \in \Low$ changes to $\High$, this affects all length-$2$ paths
where $v$ is the central, low-degree vertex.
For each pair of distinct neighbors $x, y \in \Neigh(v)$, decrease \Aux{uLv[$x, y$]} by $1$.
The running time is $\bigO({\Degree(v)}^2) = \bigO(m^{2\PartFrac})$.

Each edge may be incident to at least one low-degree vertex $v$ and form
$\bigO(m^\PartFrac)$ length-$2$ paths with the other edges incident to $v$.
The space requirement hence is $\bigO(\min(m^{1+\PartFrac},n^2))$.
By \Lemma{maintain-aux}, $\DS_{\PartFrac}$ with \Aux{uLv} can hence be
maintained in $\bigO(m^{\PartFrac} + m^{2\PartFrac}m^{-\PartFrac}) =
\bigO(m^{\PartFrac})$ amortized update time and
$\bigO(\min(m^{1+\PartFrac},n^2))$ space.
\end{proof}
\begin{lemma}\LemLabel{maintain-t}
$\DS_{\PartFrac}$ with \Aux{t} can be maintained in amortized
$\bigO(m^{\max(1-\PartFrac,\PartFrac)})$ time per update and
$\bigO(\min(m^{1+\PartFrac},n^2))$ space.
\end{lemma}
\begin{proof}
Let $\Edge{u}{v}$ be the newly inserted edge.
For each $h \in \High$, increment $\Aux{t[h]}$ by one if $h$ is adjacent to
both $u$ and $v$.
If $u \in \High$ ($v \in \High$):
Increment $\Aux{t[u]}$ ($\Aux{t[v]}$) by one for each $h \in \High$ that is
adjacent to both $u$ and $v$, and
increment $\Aux{t[u]}$ ($\Aux{t[v]}$) by \Aux{uLv[$u, v$]}.
This takes $\bigO(\Card{\High}) = \bigO(m^{1-\PartFrac})$ time.

If a vertex $v \in \Low$ changes to $\High$,
then for each pair of distinct neighbors $x, y \in \Neigh(v)$ such that
$\Edge{x}{y}\in E$,
we increase \Aux{t[$v$]} by one.
Otherwise, if $v$ changes from $\High$ to $\Low$,
set \Aux{t[$v$]}${} := 0$.
The running time is $\bigO({\Degree(v)}^2) = \bigO(m^{2\PartFrac})$.

By \Lemma{maintain-uLv}, $\DS_{\PartFrac}$ with \Aux{uLv} can be maintained in
amortized $\bigO(m^{\PartFrac})$ update time and
$\bigO(\min(m^{1+\PartFrac},n^2))$ space.
As $\Card{\High} \in \bigO(m^{1-\PartFrac})$,
the space requirement for \Aux{t} is
$\bigO(\min(m^{1+\PartFrac},n^2))$.
By \Lemma{maintain-aux}, $\DS_{\PartFrac}$ with \Aux{t} can be maintained in
$\bigO(m^{1-\PartFrac} + m^{2\PartFrac}m^{-\PartFrac}) + \bigO(m^{\PartFrac}) =
\bigO(m^{\max(1-\PartFrac,\PartFrac)})$ amortized update time and
$\bigO(\min(m^{1+\PartFrac},n^2))$ space.
\end{proof}
\begin{lemma}\LemLabel{maintain-uLLv}
$\DS_{\PartFrac}$ with \Aux{uLLv} can be maintained in amortized
$\bigO(m^{2\PartFrac})$ time per update and
$\bigO(\min(m^{1+2\PartFrac},n^2))$ space.
\end{lemma}
\begin{proof}
Let $\Edge{u}{v}$ be the newly inserted edge.
If $u \in \Low$ ($v \in \Low$),
we count the length-$3$ paths starting/ending with $\Edge{u}{v}$ as follows:
For each low-degree neighbor $w \in \NeighEx{v}(u) \cap \Low$ ($w \in \NeighEx{u}(v) \cap
\Low$), increment \Aux{uLLv[$v, x$]} (\Aux{uLLv[$u, x$]}) by one for each $x
\in \Neigh(w) \setminus \Set{u,v}$.
This takes $\bigO(m^{2\PartFrac})$ time.
If both $u \in \Low$ and $v \in \Low$, we additionally count the length-$3$ paths
having $\Edge{u}{v}$ as centerpiece in
$\bigO({\Degree(u)}^2) = \bigO(m^{2\PartFrac})$ time:
For each pair of distinct vertices $x, y$ with $x \in \NeighEx{v}(u)$, $y \in
\NeighEx{u}(v)$, increment \Aux{uLLv[$x, y$]} by one.

If a vertex $v \in \Low$ changes to $\High$,
we iterate over all pairs of distinct vertices $y, w$, where $y \in \NeighEx{v}(x)$
for some low-degree neighbor $x \in \Neigh(v) \cap \Low$ and $w \in
\NeighEx{x}(v)$, and decrease \Aux{uLLv[$w, y$]} by one.
As $v$ has $\bigO(m^{2\PartFrac})$ pairs of neighbors and each low-degree
neighbor has in turn $\bigO(m^\PartFrac)$ neighbors, the
running time is in $\bigO(m^{3\PartFrac})$.

Each edge may be incident to two low-degree vertices and hence form
$\bigO(m^{2\PartFrac})$ length-$3$ paths with the other edges incident to
the end vertices.
The space requirement hence is $\bigO(\min(m^{1+2\PartFrac},n^2))$.
By \Lemma{maintain-aux}, $\DS_{\PartFrac}$ with \Aux{uLLv} can be maintained in
$\bigO(m^{2\PartFrac} + m^{3\PartFrac}m^{-\PartFrac}) =
\bigO(m^{2\PartFrac})$ amortized update time and
$\bigO(\min(m^{1+2\PartFrac},n^2))$ space.
\end{proof}
\begin{lemma}\LemLabel{maintain-cLV}
$\DS_{\PartFrac}$ with \Aux{cLV} can be maintained in amortized
$\bigO(m^{2\PartFrac})$ time per update and $\bigO(n\min(n,m^{2\PartFrac}))$
space.
\end{lemma}
\begin{proof}
Let $\Edge{u}{v}$ be the newly inserted edge.
If $u \in \Low$ ($v \in \Low$),
we update the number of claws where $u$ ($v$) is the central vertex as follows:
For each pair of distinct neighbors $x, y \in \NeighEx{v}(u)$
($x, y \in \NeighEx{u}(v)$), increment \Aux{cLV[$x, y$]} by
one.
This accommodates for the claws where $v$ ($u$) is not an anchor vertex
and takes %
$\bigO(m^{2\PartFrac})$ time.
For the other case,
increment \Aux{cLV[$v, w$]} (\Aux{cLV[$u, w$]}) by
$\Degree(u) - 2$ ($\Degree(v) - 2$)
for each $w \in \NeighEx{v}(u)$ ($w \in \NeighEx{u}(v)$)
in %
$\bigO(m^{\PartFrac})$ time.

If a vertex $v \in \Low$ changes to $\High$,
we decrease \Aux{cLV[$x, y$]} by $\Degree(v)-2$
for each pair of distinct neighbors $x, y \in \Neigh(v)$
in total $\bigO({\Degree(v)}^2) = \bigO(m^{2\PartFrac})$ time.

As each vertex may be adjacent to at least one low-degree vertex and we store
the count for all pairs, the space requirement is in $\bigO(n^2)$.
On the other hand, each low-degree vertex has at most $\bigO(m^{2\PartFrac})$
neighbors that can serve as anchors, which yields a space requirement of
$\bigO(nm^{2\PartFrac})$.
By \Lemma{maintain-aux}, $\DS_{\PartFrac}$ with \Aux{cLV} can be maintained in
$\bigO(m^{2\PartFrac} + m^{2\PartFrac}m^{-\PartFrac}) = \bigO(m^{2\PartFrac})$
amortized update time and $\bigO(n\min(n,m^{2\PartFrac}))$ space.
\end{proof}
\begin{lemma}\LemLabel{maintain-pLL}
$\DS_{\PartFrac}$ with \Aux{pLL} can be maintained in amortized
$\bigO(m^{2\PartFrac})$ time per update and $\bigO(\min(n^2,nm^{2\PartFrac}))$
space.
\end{lemma}
\begin{proof}
Let $\Edge{u}{v}$ be the newly inserted edge.

If $u \in \Low$ ($v \in \Low$):
First, we update all paws where $u$ ($v$) is the central vertex and $v$ ($u$) is the
anchor vertex at the arm:
For each ordered pair of distinct neighbors
$x, y \in \NeighEx{v}(u)$
($x, y \in \NeighEx{u}(v)$) such that
$x \in \Low$ and $\Edge{x}{y} \in E$, increment \Aux{pLL[$v, y$]} (\Aux{pLL[$u, y$]}) by one.
This can be done in %
$\bigO(m^{2\PartFrac})$ time.
Second, we update all paws where $u$ ($v$) is the central vertex and $v$ ($u$) is
the anchor vertex in the triangle:
For each ordered pair of distinct neighbors
$x, y \in \NeighEx{v}(u)$
($x, y \in \NeighEx{u}(v)$) such that
$x \in \Low$ and $\Edge{x}{v} \in E$ ($\Edge{x}{u} \in E$),
increment \Aux{pLL[$v, y$]} (\Aux{pLL[$u, y$]}) by one.
This again can be done in %
$\bigO(m^{2\PartFrac})$ time.
Third, we update all paws where $u$ ($v$) is the non-anchor, non-central vertex
in the triangle and $v$ ($u$) is the anchor vertex in the triangle:
For each neighbor
$x \in \NeighEx{v}(u)$
($x \in \NeighEx{u}(v)$)
with $x \in \Low$ and $\Edge{v}{x} \in E$ ($\Edge{u}{x} \in E$), increment
\Aux{pLL[$v, y$]} (\Aux{pLL[$u, y$]}) by one for each $y \in \Neigh(x)
\setminus \Set{u, v}$.
The running time is in %
$\bigO(m^{2\PartFrac})$, as $u, x \in \Low$.

If both $u \in \Low$ and $v \in \Low$, we update all paws where
$\Edge{u}{v}$ connects the central vertex to the non-anchor vertex in the
triangle:
For each ordered pair of distinct neighbors $x, y \in \NeighEx{v}(u)$
with $\Edge{x}{v} \in E$
and each ordered pair of distinct neighbors $x, y \in \NeighEx{u}(v)$
with $\Edge{x}{u} \in E$, increment \Aux{pLL[$x, y$]} by one.
The running time is in $\bigO({\Degree(u)}^2 + {\Degree(v)}^2) =
\bigO(m^{2\PartFrac})$.

If a vertex $v \in \Low$ changes to $\High$:
For all paws where $v$ was the central vertex,
we iterate over all
unordered pairs of neighbors $y, z \in \Neigh(v)$
and every neighbor $x \in \Neigh(v) \cap \Low$
such that $y \neq x \neq z$.
If $\Edge{x}{y} \in E$, we decrease \Aux{pLL[$y, z$]} by one, and if
$\Edge{x}{z} \in E$, we also decrease \Aux{pLL[$y, z$]} by one.
For all paws where $v$ was the low-degree, non-central vertex in the triangle, we iterate
over all pairs of distinct neighbors $x, y \in \Neigh(v)$ such that
$\Edge{x}{y} \in E$ and $x \in \Low$, and, for each $z \in \Neigh(x) \setminus
\Set{v, y}$, decrease \Aux{pLL[$y, z$]} by one.
In this case, $\Edge{x}{z}$ forms the arm.
As $\Degree(x) \in \bigO(m^\PartFrac)$ in the second case,
the total running time is $\bigO({\Degree(v)}^3 + {\Degree(v)}^2 \cdot m^\PartFrac) =
\bigO(m^{3\PartFrac})$.

The argument for the space requirement is the same as for to \Aux{cLV}
and $\bigO(n\min(n,m^{2\PartFrac}))$ by \Lemma{maintain-cLV}.
By \Lemma{maintain-aux}, $\DS_{\PartFrac}$ with \Aux{pLL} can be maintained in
$\bigO(m^{2\PartFrac} + m^{3\PartFrac}m^{-\PartFrac}) = \bigO(m^{2\PartFrac})$
amortized update time and $\bigO(n\min(n,m^{2\PartFrac}))$ space.
\end{proof}
\begin{lemma}\LemLabel{maintain-uHv}
$\DS_{\PartFrac}$ with \Aux{uHv} can be maintained in amortized
$\bigO(m^{\max(1-\PartFrac,\PartFrac)})$ time per update and
$\bigO(n+\min(n^2,m^{2-2\PartFrac}))$ space.
\end{lemma}
\begin{proof}
Let $\Edge{u}{v}$ be the newly inserted edge.
If $u, v \in \High$, we iterate over all $h \in \High \setminus \Set{u,v}$.
If $h$ is adjacent to $u$ ($v$), increment \Aux{uHv[$h, v$]}
(\Aux{uHv[$u, h$]}), respectively, by one.
The running time is $\bigO(\Card{\High}) = \bigO(m^{1-\PartFrac})$.

If a vertex $v \in \Low$ changes to $\High$ (analogously vice-versa):
For each pair of distinct high-degree neighbors $x, y \in \Neigh(v) \cap \High$,
increment \Aux{uHv[$x, y$]} by one in total
$\bigO({\Degree(v)}^2) = \bigO(m^{2\PartFrac})$ time.
\emph{Only if} $v$ changes from $\Low$ to $\High$:
For every high-degree neighbor $w \in \Neigh(v) \cap \High$, we iterate over
all $h \in \High$ and increase \Aux{uHv[$v, h$]} by one if $\Edge{w}{h}\in E$
in total $\bigO(\Degree(v) \cdot \Card{\High}) = \bigO(m^{\PartFrac} \cdot
m^{1-\PartFrac}) = \bigO(m)$ time.
\emph{Only if} $v$ changes from $\High$ to $\Low$, we set \Aux{uHv[$v, h$]}${} := 0$
for each $h \in \High$ in total $\bigO(\Card{\High}) = \bigO(m^{1-\PartFrac})$
time.
The overall time is hence $\bigO(m^{\max(2\PartFrac,1)})$.

There are $\bigO(\min(n, m^{1-\PartFrac}))$ high-degree vertices,
which results in
$\bigO(\min(n^2,m^{2-2\PartFrac}))$ pairs of
anchor vertices.
$\DS_{\PartFrac}$ with \Aux{uHv} can hence be maintained in
$\bigO(m^{1-\PartFrac} + m^{\max(2\PartFrac,1)}m^{-\PartFrac}) =
\bigO(m^{\max(1-\PartFrac,\PartFrac)})$
amortized update time and $\bigO(n+\min(n^2,m^{2-2\PartFrac}))$ space
by \Lemma{maintain-aux}.
\end{proof}
\begin{lemma}\LemLabel{maintain-cL}
$\DS_{\PartFrac}$ with \Aux{cL} can be maintained in amortized
$\bigO(m^{2\PartFrac})$ time per update and
$\bigO(\min(n^3,nm^{3\PartFrac},m^{1+2\PartFrac}))$ space.
\end{lemma}
\begin{proof}
Let $\Edge{u}{v}$ be the newly inserted edge.
If $u \in \Low$ ($v \in \Low$),
increment \Aux{cL[$v, x, y$]} (\Aux{cL[$u, x, y$]}) by one for each pair of
distinct neighbors $x,y \in \NeighEx{v}(u)$ ($x,y \in
\NeighEx{u}(v)$).
This takes $\bigO(m^{2\PartFrac})$ time.

If a vertex $v \in \Low$ changes to $\High$:
For each triple of distinct neighbors $x, y, z \in \Neigh(v)$, we decrease
\Aux{cL[$x, y, z$]} by one in total $\bigO({\Degree(v)}^3) =
\bigO(m^{3\PartFrac})$ time.

As each edge may be incident to a low-degree vertex $v$,
the number of triples with non-zero count for \Aux{cL} is in
$\bigO(\min(n^3,nm^{3\PartFrac},m^{1+2\PartFrac}))$.
By \Lemma{maintain-aux}, $\DS_{\PartFrac}$ with \Aux{cL} can be maintained in
$\bigO(m^{2\PartFrac} + m^{3\PartFrac-\PartFrac})
= \bigO(m^{2\PartFrac})$
amortized update time and $\bigO(\min(n^3,m^{1+2\PartFrac}))$ space.
\end{proof}
\paragraph*{Non-Induced Subgraph Counts}
We are now ready to prove \Theorem{count-noninduced} and show
for each connected subgraph on four vertices how to count it
using the data structure $\DS_{\PartFrac}$.

\begin{figure}[tb]
\centering
\includegraphics[width=\textwidth]{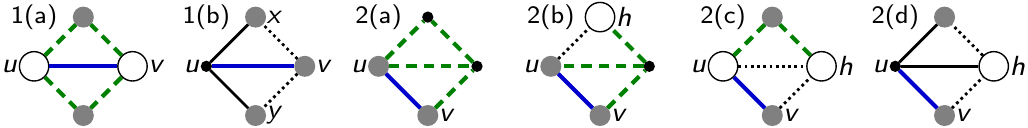}
\caption{Counting the number of diamonds that contain an edge $\Edge{u}{v}$.
Green, dashed edges belong to paths that are considered via auxiliary counts,
whereas dotted edges are edges whose presence is looked up by the algorithm.
As before, small and filled vertices have low degree, large and empty vertices high degree,
medium-sized and shaded vertices can have either high or low degree.
}%
\label{fig:count-diamonds}
\end{figure}

\begin{lemma}\LemLabel{count-diamond-edge}
Let $G = (V, E)$ be a dynamic graph,
$\PartFrac \in \Range{0}{1}$, and
$\Graphlet$ be the diamond~\PictoDiamond{}.
We can query
$\Count(G, \Graphlet, \Edge{u}{v})$
for an arbitrary edge $\Edge{u}{v} \in E$ in
$\bigO(\min(m^{\max(1-\PartFrac,2\PartFrac)},n^2))$ worst-case time
if we maintain the data structure $\DS_{\PartFrac}$ with auxiliary counts
\Aux{uLv}, \Aux{pLL}, \Aux{uHv}, and \Aux{cL}.
\end{lemma}
\begin{proof}
Edge $\Edge{u}{v}$ can either be the chord of the diamond or be part of the
$4$-cycle.
See \Figure{count-diamonds} for an illustration.

For the first case, where $\Edge{u}{v}$ is the chord:
(a)~If $u, v \in \High$, we can obtain the number of length-$2$
paths $p$ between $u$ and $v$ as $p = \Aux{uLv[$u, v$]} + \Aux{uHv[$u, v$]}$.
As each pair of length-$2$ paths forms a diamond with $\Edge{u}{v}$, the
total number of diamonds is $\binom{p}{2}$.
(b)~Otherwise, $\{u, v\} \cap \Low \neq \emptyset$.
\Wilog, $u \in \Low$.
We then iterate over all distinct, unordered pairs of neighbors $x, y \in
\NeighEx{v}(u)$ in $\bigO({\Degree(u)}^2) = \bigO(\min(m^{2\PartFrac},n^2))$ time.
For each such pair with $\Edge{x}{v}, \Edge{y}{v} \in E$, we count one diamond.

For the second case, where $\Edge{u}{v}$ is part of the cycle, we distinguish
between the degrees of the other two vertices.
(a)~The number of diamonds where the other two vertices have low degree
is given by \Aux{pLL[$u, v$]}.
Note that either $u$ or $v$ is incident to the chord.
(b)~The number of diamonds where the other vertex incident to the
chord has low degree and the fourth vertex has high degree
can be obtained by iterating over all $h \in \High \setminus \Set{u, v}$
in $\bigO(\Card{\High}) = \bigO(\min(m^{1-\PartFrac}, n))$ time.
If either $\Edge{h}{u} \in E$ or $\Edge{h}{v} \in E$, we have \Aux{cL[$u, v,
h$]} more diamonds.
If both $\Edge{h}{u}, \Edge{h}{v}\in E$, we add $2\Aux{cL[$u, v, h$]}$ to the
number of diamonds.
(c, d)~%
The number of diamonds where
the other vertex incident to the chord has high degree
can be obtained as follows:
(c)~%
If $u \in \High$ ($v \in \High$),
the number of diamonds where the chord is
incident to $u$ ($v$)
can be obtained
by iterating over all $h \in \High \setminus \Set{u, v}$
in $\bigO(\Card{\High}) = \bigO(\min(m^{1-\PartFrac},n))$ time.
For each such vertex $h$, we check whether $\Edge{u}{h}, \Edge{v}{h} \in E$ and
add
$\Aux{uLv[$u, h$]} + \Aux{uHv[$u, h$]} - 1$
($\Aux{uLv[$v, h$]} + \Aux{uHv[$v, h$]} - 1$) to the count.
The correction by $1$ is necessary because
the auxiliary counts also contain the path $\Sequence{\Edge{u}{v}, \Edge{v}{h}}$
($\Sequence{\Edge{h}{u}, \Edge{u}{v}}$).
(d)~If $u \in \Low$ ($v \in \Low$),
we iterate over all high-degree neighbors
$h \in \NeighEx{v}(u) \cap \High$
($h \in \NeighEx{u}(v) \cap \High$)
and in each case over all
$x \in \Neigh(u) \setminus \{v, h\}$
($x \in \Neigh(v) \setminus \{u, h\}$)
in total $\bigO(\min(m^{2\PartFrac},n^2))$ time and count one diamond each if
$\{h, x\}, \{h, v\} \in E$
($\{h, x\}, \{h,u\} \in E$).
\end{proof}
\begin{lemma}\LemLabel{count-diamond}
Let $G = (V, E)$ be a dynamic graph and $\Graphlet$ be the
diamond~\PictoDiamond{}.
We can maintain $\Count(G, \Graphlet)$ in amortized $\bigO(m^{2/3})$ update
time and $\bigO(\min(nm,m^\frac{5}{3}))$ space.
We can query $\Count(G, \Graphlet, e)$ for an arbitrary edge $e \in E$
in worst-case $\bigO(m^{2/3})$ time.
\end{lemma}
\begin{proof}
After an edge $\Edge{u}{v}$ was inserted or before an edge $\Edge{u}{v}$ is removed,
the number of diamonds containing it can be obtained in
$\bigO(\min(m^{\max(1-\PartFrac,2\PartFrac)},n^2))$ time
worst-case time
by \Lemma{count-diamond-edge} if $\DS_{\PartFrac}$ with auxiliary counts
\Aux{uLv}, \Aux{pLL}, \Aux{uHv}, and \Aux{cL}
is maintained.
By \Lemma{maintain-uLv},
\Lemma{maintain-pLL},
\Lemma{maintain-uHv},
and \Lemma{maintain-cL},
this can be done
in amortized $\bigO(m^{\PartFrac} +
m^{\max(1-\PartFrac,\PartFrac)} + m^{2\PartFrac}) =
\bigO(m^{\max(1-\PartFrac,2\PartFrac)})$ time and
$\bigO(\min(n^3,\max(n^2,nm^{3\PartFrac},m^{2-2\PartFrac},m^{1+2\PartFrac})))$ space.
Together with the cost for the query, this yields a total amortized update time
of $\bigO(m^{\max(1-\PartFrac,2\PartFrac)}) = \bigO(m^\frac{2}{3})$
for $\PartFrac = \frac{1}{3}$
and $\bigO(\min(nm, m^\frac{5}{3}))$ space.
By \Lemma{count-diamond-edge}, the worst-case time to query $\Count(G,
\Graphlet, e)$ for an arbitrary edge $e \in E$ then is $\bigO(m^{2/3})$.
\end{proof}

\paragraph*{Queries with Vertices and Edges and Non-Induced $s$-Subgraph Counts}
With similar techniques, we can count non-induced
triangles containing a specified vertex or edge as well as
maintain $s$-subgraph counts for patterns with up to four vertices.
\begin{theorem}\ThmLabel{count-triangles-extended}
Let $G = (V, E)$ be a dynamic graph, $\Graphlet$ be the $3$-cycle~\PictoThreeCycle{},
and $\PartFrac \in \Range{0}{1}$.
We can query $\Count(G, \Graphlet, a)$
for an arbitrary vertex or edge $a$ in
\begin{enumerate}[(i)]
\item
worst-case
$\bigO(\min(m^{2\PartFrac},n^2))$
time
with
$\bigO(m^{\max(\PartFrac,1-\PartFrac)})$
amortized update time and\\
$\bigO(\min(n^2, m^{1+\PartFrac}))$
space if $a \in V$,
\item
worst-case time
$\bigO(\min(m^{1-\PartFrac},n))$
with
$\bigO(m^{\PartFrac})$
amortized update time and\\
$\bigO(\min(n^2, m^{1+\PartFrac}))$ space
if $a \in E$.
\end{enumerate}
\end{theorem}
\begin{corollary}\CorLabel{count-triangles}
Let $G = (V, E)$ be a dynamic graph and $\Graphlet$ the
$3$-cycle~\PictoThreeCycle{}.
We can maintain $\Count(G, \Graphlet)$
with an amortized update time of $\bigO(\sqrt{m})$
and $\bigO(\min(n^2, m^{1.5}))$ space
and query
$\Count(G, \Graphlet, e)$ for $e \in E$
arbitrary in worst-case $\bigO(\sqrt{m})$ time.
We can query
$\Count(G, \Graphlet, v)$
for arbitrary $v \in V$
in worst-case $\bigO(m^{2/3})$ time
with an amortized update time of $\bigO(m^{2/3})$
and $\bigO(\min(n^2, m^{4/3}))$ space.
\end{corollary}
\begin{theorem}\ThmLabel{count-noninduced-s}
Let $G = (V, E)$ be a dynamic graph, $s \in V$, and $\Graphlet$ %
be a connected subgraph.
We can maintain the non-induced $s$-subgraph count $\Count(G, \Graphlet, s)$
in
\begin{enumerate}[(i)]
\item
worst-case constant update time
and constant space
if $\Graphlet$ is the claw \PictoClaw,
\item
amortized update time
$\bigO(\sqrt{m})$
and $\bigO(n)$ space
if $\Graphlet$ is the $3$-cycle~\PictoThreeCycle{}
or the length-$3$ path~\PictoThreePath{},
\item
amortized update time
$\bigO(m^{2/3})$
and $\bigO(n^2)$ space
if $\Graphlet$ is the paw~\PictoPaw{},
the $4$-cycle~\PictoFourCycle{}, or
the diamond~\PictoDiamond{}, %
\item
worst-case update time $\bigO(m)$
and constant space
if $\Graphlet$ is
the $4$-clique~\PictoClique{}.
\end{enumerate}
\end{theorem}
\section{Lower Bounds}\label{sect:lowerbounds}
We give new lower bounds for detecting and counting induced and non-induced subgraphs.
\paragraph*{Induced Subgraph Counts}
Our results for counting induced subgraphs on four vertices are
conditioned on the combinatorial $k$-clique conjecture:
\begin{theorem}\ThmLabel{count-induced}
Let $G$ be a dynamic graph, $\Graphlet \in \GraphletSet_4$,
and let $\gamma > 0$ be a small constant.
There is no incremental or fully dynamic combinatorial algorithm with preprocessing time $\bigO(m^{2-\gamma})$ for maintaining
$\ICount(G, \Graphlet)$ in amortized update time $\bigO(m^{1-\gamma})$ and query time $\bigO(m^{2-\gamma})$, unless the $k$-clique conjecture fails.
\end{theorem}
Before we turn to the proof, we recall %
the following relations between subgraph counts.
\begin{lemma}[\cite{DBLP:journals/tcs/EppsteinGST12}]\LemLabel{subgraph-multiplicities}
For each pair $P, P' \in \GraphletSet_4$, %
the non-induced subgraph count
$\Count(P, P') = 1$ if $P = P'$
and otherwise
nonzero only in the following cases:
\setlength{\abovedisplayskip}{1pt}
\setlength{\belowdisplayskip}{0pt}
\begin{align*}
\Count(\PictoPaw{},    \PictoThreePath) &= 2,  &
\Count(\PictoFourCycle,\PictoThreePath) &= 4,  &
\Count(\PictoDiamond,  \PictoThreePath) &= 6,  &
\Count(\PictoClique{}, \PictoThreePath) &= 12, \\
\Count(\PictoPaw{},    \PictoClaw{})    &= 1,  &
\Count(\PictoDiamond,  \PictoClaw{})    &= 2,  &
\Count(\PictoClique{}, \PictoClaw{})    &= 4,  &
\Count(\PictoDiamond,  \PictoPaw{})     &= 4,  \\
\Count(\PictoClique{}, \PictoPaw{})     &= 12, &
\Count(\PictoDiamond,  \PictoFourCycle) &= 1,  &
\Count(\PictoClique{}, \PictoFourCycle) &= 3,  &
\Count(\PictoClique{}, \PictoDiamond)   &= 6.  %
\end{align*}
\end{lemma}
\begin{proposition}[\cite{eq}]\PropLabel{counts-induced-and-not}%
Let $G, \Graphlet$ be graphs and let $k$ be the number of vertices of $\Graphlet$.
Then,
$\Count(G, \Graphlet) = \sum_{P \in \GraphletSet_k}\ICount(G, P) \cdot \Count(P, \Graphlet)$.
\end{proposition}
\begin{lemma}\LemLabel{induced-plugin}
Let $G = (V, E)$ be a graph.
The following relationships between induced and non-induced subgraph counts hold:
\setlength{\abovedisplayskip}{1pt}
\setlength{\belowdisplayskip}{0pt}
\begin{align*}
\ICount(G, \PictoClique) &= \Count(G, \PictoClique) \\
\ICount(G, \PictoDiamond) &= \Count(G, \PictoDiamond) - 6\ICount(G, \PictoClique) \\
\ICount(G, \PictoFourCycle) &= \Count(G, \PictoFourCycle) -   \Count(G, \PictoDiamond) +  3\ICount(G, \PictoClique)\\
\ICount(G, \PictoPaw) &= \Count(G, \PictoPaw) - 4 \Count(G, \PictoDiamond) + 12\ICount(G, \PictoClique)\\
\ICount(G, \PictoClaw) &= \Count(G, \PictoClaw) -   \Count(G, \PictoPaw) +  2 \Count(G, \PictoDiamond) - 4\ICount(G, \PictoClique)\\
\ICount(G, \PictoThreePath) &= \Count(G, \PictoThreePath) - 2 \Count(G, \PictoPaw) -  4 \Count(G, \PictoFourCycle) + 6\Count(G, \PictoDiamond) - 12\ICount(G, \PictoClique)
\end{align*}
\end{lemma}
\begin{proof}
Let $\Graphlet \in \GraphletSet_4$.
The statement follows from \Proposition{counts-induced-and-not} and
\Lemma{subgraph-multiplicities} by substituting equations and solving them
for $\ICount(G, \Graphlet)$ in the order as listed.
\end{proof}

\begin{corollary}\CorLabel{count-induced}
Let $G = (V, E)$ be a dynamic graph and $\Graphlet \in \GraphletSet_4$.
We can maintain $\ICount(G, \Graphlet)$ with $\bigO(m)$
worst-case
update time and %
constant space.
\end{corollary}

\begin{proof}[Proof of \Theorem{count-induced}]
First consider the case that $\Graphlet = \PictoClique$.
Suppose there is an incremental or fully dynamic algorithm $\Alg$ that
maintains $\ICount(G, \Graphlet)$ in time $\bigO(m^{1-\gamma})$
with query time $\bigO(m^{2-\gamma})$
for some $\gamma> 0$.
Construct an algorithm $\Alg'$ for static $4$-clique
detection as follows:
Run $\Alg$ on an initially empty graph, insert all edges
one-by-one in total $\bigO(m^{2-\gamma})$ time,
and query the result in $\bigO(m^{2-\gamma})$ time.
As $\bigO(m^{2-\gamma}) = \bigO(n^{4-2\gamma})$ this
contradicts \Conjecture{k-clique}.

For the remaining five induced four-vertex subgraphs, let $\Graphlet$ be such a
subgraph.
We construct a deterministic algorithm $\Alg'$ for static $4$-clique
detection as follows:
$\Alg'$ executes the above operations for our non-induced subgraph counting algorithm from
\Section{algorithms}, which can maintain the number of all connected subgraphs
on four vertices with $\bigO(m^{2/3})$ amortized update time and $\bigO(1)$ query time by
\Theorem{count-noninduced}.
It thus takes $\bigO(m^{5/3})$ time in total to compute $\Count(G,
\Graphlet')$ for all $\Graphlet' \in \GraphletSet_4 \setminus \{ \PictoClique \}$.
Assume by contradiction that there exists an algorithm $\Alg^*$ that
maintains $\ICount(G, \Graphlet)$ in update time $\bigO(m^{1-\gamma})$
and query time $\bigO(m^{2-\gamma})$.
Then $\Alg'$ also executes the same operations with $\Alg^*$ to compute $\ICount(G, \Graphlet)$.
Using the formula for $\Graphlet$ in \Lemma{induced-plugin}, $\Alg'$ can
solve the static 4-clique detection problem in time $\bigO(m^{\max(5/3,
2-\gamma)})
\subseteq \bigO(n^{4-\delta})$ time for
$\delta = \min(2\gamma,2/3) > 0$, a contradiction to \Conjecture{k-clique}.
\end{proof}

\paragraph*{Non-Induced Subgraph Counts}
In this section, we give new lower bounds for detecting (and thus counting)
cycles of arbitrary length, paws, diamonds, and $4$-cliques\footnote{\Theorem{count-induced} applies also to $4$-cliques, but it is based on a different assumption.},
as well as counting length-$3$ paths.
Our results are based on the \OMv{} conjecture.
In~\cite{henzinger2015unifying} it is proven that instead of reducing from
\OMv{} directly it suffices to reduce from the following \oneuMv{} version:
For any positive integer parameters $n_1$, $n_2$, given an $n_1 \times n_2$
matrix $M$, there is no algorithm with preprocessing time polynomial in $n_1$ and $n_2$
that computes
for an $n_1$-dimensional vector $u$ and an $n_2$-dimensional vector
$v$  the product
$u^\top M v$
in time $\bigO(n_1 n_2^{1 -\delta} + n_1^{1 -\delta} n_2)$ for any
small constant $\delta > 0$
with error probability of at most $\frac{1}{3}$ in the word-RAM
model with $\bigO(\log n)$ bit words.
\begin{theorem}\ThmLabel{lowerbounds}
Let $G$ be a partially dynamic graph and let $\Graphlet$ be
the paw~\PictoPaw{},
the diamond~\PictoDiamond{},
the $4$-clique~\PictoClique, or
a $k$-cycle with $k \geq 3$.
On condition of \Conjecture{omv}, there is no partially dynamic algorithm to
maintain whether $\Count(G, \Graphlet) > 0$ with polynomial preprocessing time
and worst-case update time $\bigO(m^{1/2-\delta})$
and query time $\bigO(m^{1-\delta})$
with an error probability of at most $1/3$
for any $\delta > 0$.
This also holds for fully dynamic algorithms with amortized update time
and for paws, diamonds, $4$-cliques, or odd
$k$-cycles containing a specific vertex $s$,
as well as for maintaining the number of length-$3$ paths~\PictoThreePath{}
and the number of length-$3$ paths containing a specific vertex $s$.
\end{theorem}
\begin{figure}[tb]
\centering
\def\PHeight{2.2cm}
\def\Shift{.3cm}
\begin{tikzpicture}[inner sep=0pt]
\node[font=\small,align=left] (mat) {%
$M = \Bigl(
\begin{smallmatrix}
1 & 0 & 1 & 0 \\
0 & 1 & 1 & 0 \\
0 & 1 & 1 & 1 \\
\end{smallmatrix}%
\Bigr)$\\
$u^\top = (\begin{smallmatrix}1&1&0   \end{smallmatrix})$\\
$v^\top = (\begin{smallmatrix}0&1&1&0 \end{smallmatrix})$
};
\node[anchor=west,xshift=\Shift] (pic) at (mat.east) {\includegraphics[height=\PHeight]{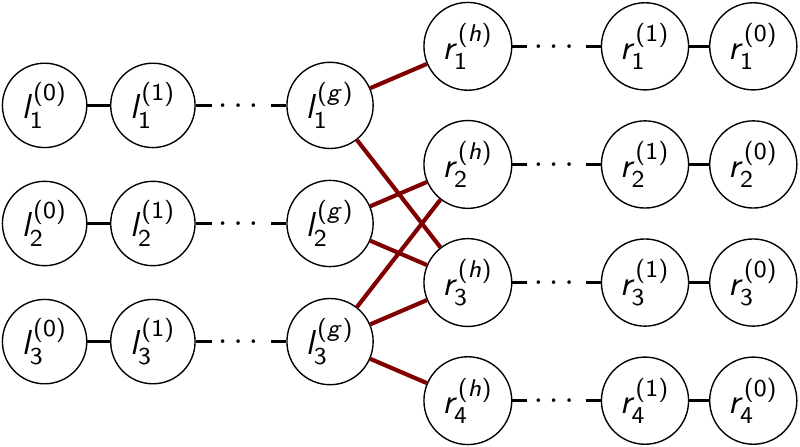}};
\node[anchor=west,xshift=\Shift] (cycle) at (pic.east) {\includegraphics[height=\PHeight]{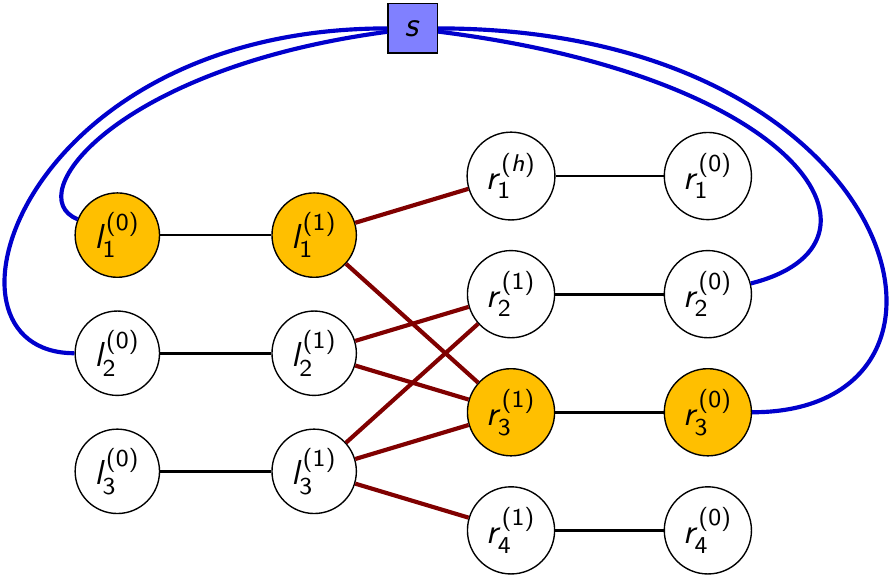}};
\node[anchor=west,xshift=\Shift] (diamond) at (cycle.east) {\includegraphics[height=\PHeight]{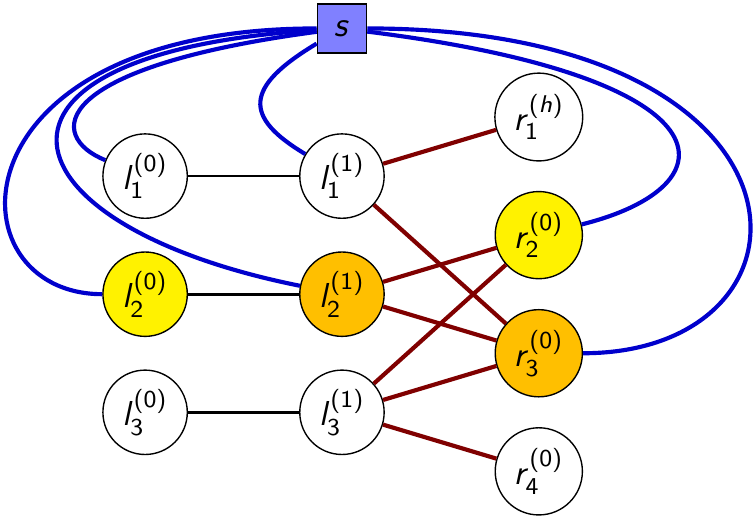}};
\end{tikzpicture}
\caption{Construction of $G_{M,g,h}$
and $G$ for $u^\top M v$ for
($s$)-$5$-cycle and
($s$)-diamond
detection.}%
\FigLabel{lb-construct-G}\FigLabel{lb-examples}
\end{figure}

Our constructions build on the following graph (see \Figure{lb-construct-G} for
an example).
\begin{definition}[$G_{M,g,h}$]\DefLabel{GMgh}
Given a matrix $M \in {\{0,1\}}^{n_1\times n_2}$ and two integers
$g, h \geq 0$, we denote by
$G_{M,g,h} =
(\bigcup\limits_{\mathclap{0 \leq p \leq g}} L^{(p)} \cup\,
\bigcup\limits_{\mathclap{0 \leq q \leq h}} R^{(q)},
E_L \cup E_R \cup E_M)$
the $(g+h+2)$-partite graph with
\setlength{\abovedisplayskip}{1pt}
\setlength{\belowdisplayskip}{1pt}
\begin{align*}
L^{(p)} &=\big\{l^{(p)}_1,\dots, l^{(p)}_{n_1}\big\},\quad 0 \leq p \leq g; &
E_L &= \big\{(l^{(p)}_{i}, l^{(p+1)}_{i}) \mid 1 \leq i \leq n_1 \wedge 0 \leq p < g \big\} \\
R^{(q)} &=\big\{r^{(q)}_1,\dots, r^{(q)}_{n_2}\big\},\quad 0 \leq q \leq h; &
E_R &= \big\{(r^{(q)}_{j}, r^{(q+1)}_{j}) \mid 1 \leq j \leq n_2 \wedge 0 \leq q < h \big\}
\end{align*}
and
$E_M = \{(l^{(g)}_{i}, r^{(h)}_j) \mid M_{ij} = 1 \}$.
$G_{M,g,h}$ has $(g+1) \cdot n_1 + (h+1) \cdot n_2$ vertices and at most
$n_1n_2 + g\cdot n_1 + h\cdot n_2$ edges.
All vertices $x$ in $L^{(p)}$ for $1 \leq p < g$ and in $R^{(q)}$ for $1 \leq q < h$ have $\Degree(x)=2$, %
whereas all vertices $y$ in $L^{(0)}$ and in $R^{(0)}$ have $\Degree(y)=1$. %

For convenience, we set
$L := L^{(0)}$,
$R := R^{(0)}$,
$l_i := l^{(0)}_i$, and
$r_j := r^{(0)}_j$.
\end{definition}
\begin{observation}\ObLabel{cycle-bip-even}
Every cycle in $G_{M,g,h}$ is even and has length at least $4$.
\end{observation}

Let $s$ be a fixed vertex in the graph. The $s$-$k$-cycle detection problem requires the algorithm to detect whether
a $k$-cycle containing $s$ exists. We use the same notation for the other subgraphs as well.
\begin{lemma}%
\LemLabel{subgraph-det-cnt}
Given a partially dynamic algorithm $\Alg$ for one of the problems
listed below, one can solve \oneuMv{} with parameters $n_1$ and $n_2$ by running the
preprocessing step of $\Alg$ on a graph with $\bigO(m + \sqrt{m}\cdot k)$ edges
and $\Theta(\sqrt{m}\cdot k)$ vertices, and then making $\bigO(\sqrt{m})$
insertions (or $\bigO(\sqrt{m})$ deletions) and $1$ query, where $m$ is such
that $n_1=n_2=\sqrt{m}$.
The problems are

\smallskip
\begin{enumerate}[(a)]
\begin{minipage}{0.47\textwidth}
\item ($s$-)$k$-Cycle Detection for odd $k$
\item ($s$-)Paw Detection
\item\label{itm:diamond-det} ($s$-)Diamond Detection
\end{minipage}
\begin{minipage}{0.47\textwidth}
\item\label{itm:4-clique-det} ($s$-)$k$-Clique Detection for $k = 4$
\item ($s$-)length-$k$ Path Counting for $k = 3$
\item $k$-Cycle Detection
\end{minipage}
\end{enumerate}
\end{lemma}
\begin{proof}[Proof of Case~(\ref{itm:diamond-det})]
We only prove the decremental case. Consider a \oneuMv{} problem with
$n_1 = n_2 = \sqrt m$.
Given $M$, we construct the tripartite graph $G$ from $G_{M,1,0}$
by adding to it a vertex $s$
and connecting it by an edge to every vertex in $G_{M,1,0}$.
Thus, the total number of edges is at most
$n_1n_2 + 3n_1 + n_2=\bigO(m)$.
Once $u$ and $v$ arrive, we delete
$\Edge{s}{l_i}$ and $\Edge{s}{l^{(1)}_i}$ iff $u_i=0$ and
delete $\Edge{r_j}{s}$ iff $v_j=0$.
See \Figure{lb-examples} for an example.

Consider the case that $G$ contains a diamond with chord $e$.
As every triangle must be incident to $s$ by \Observation{cycle-bip-even}, $e =
\Edge{s}{x}$ for some $x \in L \cup L^{(1)} \cup R$.
Furthermore, all vertices in $L$ have degree at most two and $x$ has degree at
least three, so $x \not\in L$.
If $x = r_j \in R$, then by construction, $x$ must have two neighbors
$l^{(1)}_i, l^{(1)}_{i'} \in L^{(1)}$ such that there are edges
$\Edge{s}{l^{(1)}_i}$ and $\Edge{s}{l^{(1)}_{i'}}$ in $G$.
Again by construction, there are thus also edges $\Edge{s}{l_i}$ and
$\Edge{s}{l_{i'}}$ and a diamond $\{s, l_i, l^{(1)}_i, x = r_j\}$ with
chord $\Edge{s}{l^{(1)}_i}$.
As each vertex in $L^{(1)}$ is adjacent to exactly one vertex in $L$, every
diamond must contain a vertex $r \in R$ and the edge $\Edge{r}{s}$.
Hence, we have $u^\top Mv= 1$
iff there is a diamond in $G$
iff there is a diamond incident to $s$.
In total, we need to do $2n_1+n_2=\bigO(\sqrt{m})$ updates and $1$ query.
\end{proof}
\section{Conclusion}\label{sect:conclusion}
Our focus in this work was especially on non-induced and induced four-vertex subgraphs.
We gave improved both upper and lower bounds for detecting or counting
four-vertex subgraphs in the dynamic setting, thereby
closing the gap (w.r.t.\ improvements by a polynomial factor) for counting
non-induced length-$3$ paths and narrowing it considerably for non-induced
paws, $4$-cycles, and diamonds.
For counting induced subgraphs, we showed that the update time of the algorithm
by Eppstein \etal{}~\cite{DBLP:journals/tcs/EppsteinGST12} cannot be improved
by a polynomial factor, but that a better space complexity can be achieved
in the worst case.

Many of our lower bounds also apply to subgraphs with more than four vertices,
but to the best of our knowledge, only algorithms for cliques have been
considered here so far.
Hence, besides closing the gap for four-vertex subgraphs, the complexity of
detecting and counting subgraphs with five or more vertices would be an
interesting field for future work.

We also investigated the complexity of querying the number of subgraphs
containing a specific edge, as such queries are relevant, \eg, to
measure the similarity of graphs via histograms.
This can similarly be done for vertices.
As a by-product of our results for four-vertex subgraphs, we showed for
$3$-cycles that such vertex queries can be answered in
$\bigO(m^\frac{2}{3})$.
The complexity of vertex queries for larger subgraphs remains an open question.

Further interesting lines to follow regard the complexity of approximate
counting, counting all approximately densest subgraphs, as well
as the complexity of \emph{enumerating} subgraphs.

Our work was strongly motivated also by the practical relevance of
counting subgraphs in the dynamic setting.
For this reason, we consider an experimental evaluation of dynamic subgraph counting
algorithms a relevant and very interesting task for future work.

\bibliography{main}

\begin{thebibliography}{10}

\bibitem{abboud2018if}
Amir Abboud, Arturs Backurs, and Virginia~Vassilevska Williams.
\newblock If the current clique algorithms are optimal, so is valiant's parser.
\newblock {\em SIAM Journal on Computing}, 47(6):2527--2555, 2018.

\bibitem{AYZ}
N.~Alon, R.~Yuster, and U.~Zwick.
\newblock Finding and counting given length cycles.
\newblock {\em Algorithmica}, 17(3):209--223, Mar 1997.
\newblock \href {https://doi.org/10.1007/BF02523189}
  {\path{doi:10.1007/BF02523189}}.

\bibitem{streaming2}
Ziv Bar-Yossef, Ravi Kumar, and D.~Sivakumar.
\newblock Reductions in streaming algorithms, with an application to counting
  triangles in graphs.
\newblock In {\em Proceedings of the Thirteenth Annual ACM-SIAM Symposium on
  Discrete Algorithms}, SODA '02, page 623–632, USA, 2002. Society for
  Industrial and Applied Mathematics.

\bibitem{DBLP:conf/soda/BergamaschiHGWW21}
Thiago Bergamaschi, Monika Henzinger, Maximilian~Probst Gutenberg,
  Virginia~Vassilevska Williams, and Nicole Wein.
\newblock New techniques and fine-grained hardness for dynamic near-additive
  spanners.
\newblock In D{\'{a}}niel Marx, editor, {\em Proceedings of the 2021 {ACM-SIAM}
  Symposium on Discrete Algorithms, {SODA} 2021, Virtual Conference, January 10
  - 13, 2021}, pages 1836--1855. {SIAM}, 2021.
\newblock \href {https://doi.org/10.1137/1.9781611976465.110}
  {\path{doi:10.1137/1.9781611976465.110}}.

\bibitem{bringmann2019fine}
Karl Bringmann, Nick Fischer, and Marvin K{\"u}nnemann.
\newblock A fine-grained analogue of schaefer’s theorem in p: Dichotomy of
  $\exists k\forall$-quantified first-order graph properties.
\newblock In {\em 34th Computational Complexity Conference}, pages 1--27.
  Schloss Dagstuhl, 2019.

\bibitem{streaming1}
Luciana~S. Buriol, Gereon Frahling, Stefano Leonardi, Alberto
  Marchetti-Spaccamela, and Christian Sohler.
\newblock Counting triangles in data streams.
\newblock In {\em Proceedings of the Twenty-Fifth ACM SIGMOD-SIGACT-SIGART
  Symposium on Principles of Database Systems}, PODS '06, page 253–262, New
  York, NY, USA, 2006. Association for Computing Machinery.
\newblock \href {https://doi.org/10.1145/1142351.1142388}
  {\path{doi:10.1145/1142351.1142388}}.

\bibitem{DLS2020}
Laxman Dhulipala, Quanquan~C. Liu, Julian Shun, and Shangdi Yu.
\newblock Parallel batch-dynamic k-clique counting.
\newblock {\em CoRR}, abs/2003.13585, 2020.
\newblock URL: \url{https://arxiv.org/abs/2003.13585}, \href
  {http://arxiv.org/abs/2003.13585} {\path{arXiv:2003.13585}}.

\bibitem{DvorakDynamicSubgraph}
Zden{\v{e}}k Dvo{\v{r}}{\'a}k and Vojt{\v{e}}ch T{\r{u}}ma.
\newblock A dynamic data structure for counting subgraphs in sparse graphs.
\newblock In Frank Dehne, Roberto Solis-Oba, and J{\"o}rg-R{\"u}diger Sack,
  editors, {\em Algorithms and Data Structures}, pages 304--315, Berlin,
  Heidelberg, 2013. Springer Berlin Heidelberg.

\bibitem{DBLP:journals/tcs/EppsteinGST12}
David Eppstein, Michael~T. Goodrich, Darren Strash, and Lowell Trott.
\newblock Extended dynamic subgraph statistics using h-index parameterized data
  structures.
\newblock {\em Theor. Comput. Sci.}, 447:44--52, 2012.
\newblock \href {https://doi.org/10.1016/j.tcs.2011.11.034}
  {\path{doi:10.1016/j.tcs.2011.11.034}}.

\bibitem{DBLP:journals/jgaa/EppsteinS12}
David Eppstein and Emma~S. Spiro.
\newblock The h-index of a graph and its application to dynamic subgraph
  statistics.
\newblock {\em J. Graph Algorithms Appl.}, 16(2):543--567, 2012.
\newblock \href {https://doi.org/10.7155/jgaa.00273}
  {\path{doi:10.7155/jgaa.00273}}.

\bibitem{full-version}
Kathrin Hanauer, Monika Henzinger, and Qi~Cheng Hua.
\newblock Fully dynamic four-vertex subgraph counting.
\newblock {\em CoRR}, abs/2106.15524, 2021.
\newblock URL: \url{https://arxiv.org/abs/2106.15524}, \href
  {http://arxiv.org/abs/2106.15524} {\path{arXiv:2106.15524}}.

\bibitem{henzinger2015unifying}
Monika Henzinger, Sebastian Krinninger, Danupon Nanongkai, and Thatchaphol
  Saranurak.
\newblock Unifying and strengthening hardness for dynamic problems via the
  online matrix-vector multiplication conjecture.
\newblock In {\em Proceedings of the Forty-Seventh Annual ACM Symposium on
  Theory of Computing}, pages 21--30, 2015.

\bibitem{HLS2022}
Monika Henzinger, Andrea Lincoln, and Barna Saha.
\newblock The complexity of average-case dynamic subgraph counting.
\newblock In {\em Proceedings of the Thirty-Third Annual {ACM-SIAM} Symposium
  on Discrete Algorithms, {SODA} 2022, Alexandria, Virginia, USA, January 9-12,
  2022}. {SIAM}, 2022.
\newblock to appear.

\bibitem{dynamicinsertions}
Ahmet Kara, Hung~Q. Ngo, Milos Nikolic, Dan Olteanu, and Haozhe Zhang.
\newblock {Counting Triangles under Updates in Worst-Case Optimal Time}.
\newblock In Pablo Barcelo and Marco Calautti, editors, {\em 22nd International
  Conference on Database Theory (ICDT 2019)}, volume 127 of {\em Leibniz
  International Proceedings in Informatics (LIPIcs)}, pages 4:1--4:18,
  Dagstuhl, Germany, 2019. Schloss Dagstuhl--Leibniz-Zentrum fuer Informatik.
\newblock URL: \url{http://drops.dagstuhl.de/opus/volltexte/2019/10306}, \href
  {https://doi.org/10.4230/LIPIcs.ICDT.2019.4}
  {\path{doi:10.4230/LIPIcs.ICDT.2019.4}}.

\bibitem{DBLP:journals/tods/KaraNNOZ20}
Ahmet Kara, Hung~Q. Ngo, Milos Nikolic, Dan Olteanu, and Haozhe Zhang.
\newblock Maintaining triangle queries under updates.
\newblock {\em {ACM} Trans. Database Syst.}, 45(3):11:1--11:46, 2020.
\newblock \href {https://doi.org/10.1145/3396375} {\path{doi:10.1145/3396375}}.

\bibitem{Kloks}
T.~Kloks, D.~Kratsch, and H.~M{\"u}ller.
\newblock Finding and counting small induced subgraphs efficiently.
\newblock In Manfred Nagl, editor, {\em Graph-Theoretic Concepts in Computer
  Science}, pages 14--23, Berlin, Heidelberg, 1995. Springer Berlin Heidelberg.

\bibitem{eq}
William~L Kocay.
\newblock Some new methods in reconstruction theory.
\newblock In {\em Combinatorial Mathematics IX}, pages 89--114. Springer, 1982.

\bibitem{mapreduce}
Tamara~G. Kolda, Ali Pinar, Todd Plantenga, C.~Seshadhri, and Christine Task.
\newblock Counting triangles in massive graphs with mapreduce.
\newblock {\em SIAM Journal on Scientific Computing}, 36(5):S48--S77, 2014.
\newblock \href {http://arxiv.org/abs/https://doi.org/10.1137/13090729X}
  {\path{arXiv:https://doi.org/10.1137/13090729X}}, \href
  {https://doi.org/10.1137/13090729X} {\path{doi:10.1137/13090729X}}.

\bibitem{LeGall14}
F.~Le~Gall.
\newblock Powers of tensors and fast matrix multiplication.
\newblock In K.~Nabeshima, K.~Nagasaka, F.~Winkler, and
  {\'{A}}.~Sz{\'{a}}nt{\'{o}}, editors, {\em International Symposium on
  Symbolic and Algebraic Computation, {ISSAC} '14, Kobe, Japan, July 23-25,
  2014}, pages 296--303. {ACM}, 2014.
\newblock \href {https://doi.org/10.1145/2608628.2608664}
  {\path{doi:10.1145/2608628.2608664}}.

\bibitem{lincoln2018tight}
Andrea Lincoln, Virginia~Vassilevska Williams, and Ryan Williams.
\newblock Tight hardness for shortest cycles and paths in sparse graphs.
\newblock In {\em Proceedings of the Twenty-Ninth Annual ACM-SIAM Symposium on
  Discrete Algorithms}, pages 1236--1252. SIAM, 2018.

\bibitem{Sahu20}
Siddhartha Sahu, Amine Mhedhbi, Semih Salihoglu, Jimmy Lin, and M.~Tamer
  {\"{O}}zsu.
\newblock The ubiquity of large graphs and surprising challenges of graph
  processing: extended survey.
\newblock {\em {VLDB} J.}, 29(2-3):595--618, 2020.
\newblock \href {https://doi.org/10.1007/s00778-019-00548-x}
  {\path{doi:10.1007/s00778-019-00548-x}}.

\end{thebibliography}
\clearpage
\appendix
\section{Omitted Proofs from \Section{algorithms}}\SectLabel{alg-proofs}
\begin{proof}[Proof of \Theorem{count-any}]
If $\Graphlet$ is the claw~\PictoClaw{}, the number of claws affected by the
insertion or deletion by an edge $\Edge{u}{v}$ can be obtained directly from
the degrees of $u$ and $v$, and the edge counts can be updated in total $\bigO(n)$
time by iterating over all $x \in \NeighEx{v}(u)$ ($y \in \NeighEx{u}(v)$).

For the other subgraphs, consider the following algorithm $\Alg$: Given an edge
$\{u, v\}$, it iterates over all edges $\Edge{x}{y}\in E$, such that $\{u, v\}
\cap \{x,y\} = \emptyset$,
tests in constant time whether $\{u, v, x, y\}$ form $\Graphlet$,
and updates any global or edge counts related to one of the up to six possible
edges in $\bigO(m)$ total time.

If $\Graphlet$ is the
$4$-cycle~\PictoFourCycle{}
or the $4$-clique~\PictoClique{},
$\Alg$ can be used to enumerate all subgraphs that 
contain an edge $\{u, v\}$.

If $\Graphlet$ is the length-$3$ path~\PictoThreePath{}, $\Alg$ only
enumerates those paths where $\{u, v\}$ is \emph{not} the centerpiece.
For the remaining length-$3$ paths, we
iterate over all $x \in \NeighEx{v}(u)$.
If $\Edge{v}{x} \in E$, the number of length-$3$ paths containing
$\Edge{u}{v}$ and $\Edge{v}{x}$ equals $\Degree(v)-2$,
and $\Degree(v)-1$ if $\Edge{v}{x} \not\in E$.
The same can be done symmetrically for all $y \in \NeighEx{u}(v)$
and takes $\bigO(n)$ time in total.
This procedure can be used to update the edge counts for all edges
that form a length-$3$ path with $\Edge{u}{v}$ as centerpiece.

If $\Graphlet$ is the paw~\PictoPaw{}, $\Alg$ only
enumerates those paws where $\Edge{u}{v}$ is either the arm
or the edge in the triangle not incident to the central vertex.
For the remaining paws where $u$ is the central vertex,
we iterate over all $x \in \NeighEx{u}(v)$ and test
whether $\Edge{x}{u} \in E$.
The number of paws containing the triangle $\Edge{u}{v}$, $\Edge{v}{x}$, and
$\Edge{x}{u}$ then equals $\Degree(u)-2$.
We store the number of such triangles $t_u$ and iterate over all $y \in
\NeighEx{v}(u)$.
If $\Edge{v}{y} \not\in E$, the number of paws containing
$\Edge{u}{v}$ and $\Edge{u}{y}$ as arm is $t_u$,
and $t_u - 1$ otherwise.
The same can be done symmetrically for the case where $v$ is the central vertex
in total $\bigO(n)$ time, and, together with $\Alg$, it can be used to update
the edge counts for all edges that form a paw with $\Edge{u}{v}$.

If $\Graphlet$ is the diamond~\PictoDiamond{}, $\Alg$ only
enumerates those diamond where $\Edge{u}{v}$ is \emph{not} the chord.
For the remaining diamonds,
we iterate over all $x \in \NeighEx{u}(v)$ and test
whether $\Edge{x}{u} \in E$.
We store the number of such triangles $t$ and repeat the same iteration.
The number of diamonds containing $\Edge{u}{v}$, $\Edge{u}{x}$,
and $\Edge{v}{x}$ then equals $t-1$.

We can use the subgraph specific additional procedures plus
algorithm $\Alg$
in all cases except the claw~\PictoClaw{}
to (i)~update $\Count(G,
\Graphlet)$ on each edge update,
(ii)~query $\Count(G, \Graphlet, \Edge{u}{v})$ without using additional precomputed information,
(iii)~store and update $\Count(G, \Graphlet, e)$ for each edge $e \in E$.
\end{proof}

\subsection{Counting Non-Induced Subgraphs}\SectLabel{app-subgraph}
\paragraph*{Counting $3$-cycles}
\begin{lemma}\LemLabel{count-triangle-edge}
Let $G = (V, E)$ be a dynamic graph, $\PartFrac \in \Range{0}{1}$,
and let $\Graphlet$ be the $3$-cycle~\PictoThreeCycle{}.
We can query
$\Count(G, \Graphlet, \Edge{u}{v})$
for an arbitrary edge $\Edge{u}{v} \in E$ in
worst-case $\bigO(\min(m^{1-\PartFrac},n))$ time
if we maintain the data structure $\DS_{\PartFrac}$ with auxiliary count
\Aux{uLv}.
\end{lemma}
\begin{proof}
The number of $3$-cycles where the third vertex has low degree
is given directly by \Aux{uLv[$u, v$]}.
To count all $3$-cycles where the third vertex has high degree, we iterate over
all $h \in \High$ and test whether $\Edge{u}{h}, \Edge{v}{h}\in E$.
The running time is hence $\bigO(\Card{\High}) = \bigO(\min(m^{1-\PartFrac},n))$.
\end{proof}
\begin{corollary}\CorLabel{query-triangle-edge}
Let $G = (V, E)$ be a dynamic graph, $\PartFrac \in \Range{0}{1}$,
and let $\Graphlet$ be the $3$-cycle~\PictoThreeCycle{}.
We can query
$\Count(G, \Graphlet, \Edge{u}{v})$
for an arbitrary edge $\Edge{u}{v} \in E$ in
worst-case
$\bigO(\min(m^{1-\PartFrac},n))$ time,
with
$\bigO(m^{\PartFrac})$
amortized update time
and $\bigO(\min(m^{1+\PartFrac},n^2))$ space.
\end{corollary}
\begin{proof}
Follows directly from \Lemma{count-triangle-edge} and
\Lemma{maintain-uLv}.
\end{proof}

\begin{lemma}\LemLabel{query-triangle-vertex}
Let $G = (V, E)$ be a dynamic graph, $\PartFrac \in \Range{0}{1}$,
and let $\Graphlet$ be the $3$-cycle~\PictoThreeCycle{}.
We can query
$\Count(G, \Graphlet, v)$
for an arbitrary vertex $v \in V$
in worst-case
$\bigO(\min(m^{2\PartFrac},n^2))$
time
if we maintain the data structure $\DS_{\PartFrac}$ with auxiliary count
\Aux{t}.
\end{lemma}
\begin{proof}
If $v \in \High$, the number of $3$-cycles containing $v$ is given
directly by \Aux{t[$v$]}.

If $v \in \Low$,
we iterate over all unordered pairs of distinct neighbors $x, y \in
\Neigh(v)$ and test whether $\Edge{x}{y}\in E$.
This takes $\bigO({\Degree(v)}^2) = \bigO(\min(m^{2\PartFrac},n^2))$ time.
\end{proof}

\begin{corollary}\CorLabel{query-triangle-vertex}
Let $G = (V, E)$ be a dynamic graph, $\PartFrac \in \Range{0}{1}$,
and let $\Graphlet$ be the $3$-cycle~\PictoThreeCycle{}.
We can query
$\Count(G, \Graphlet, v)$
for an arbitrary vertex $v \in V$
in worst-case $\bigO(\min(m^{2\PartFrac},n^2))$ time, with
$\bigO(m^{\max(\PartFrac,1-\PartFrac)})$
amortized update time
and $\bigO(\min(m^{1+\PartFrac},n^2))$ space.
\end{corollary}
\begin{proof}
Follows directly from \Lemma{query-triangle-vertex}
and
\Lemma{maintain-t}.
\end{proof}
\paragraph*{Counting Length-$3$ paths}
\begin{lemma}\LemLabel{count-3path-edge}
Let $G = (V, E)$ be a dynamic graph, $\PartFrac \in \Range{0}{1}$,
and let $\Graphlet$ be the length-$3$ path~\PictoThreePath{}.
We can query
$\Count(G, \Graphlet, \Edge{u}{v})$
for an arbitrary edge $\Edge{u}{v} \in E$ in
worst-case $\bigO(\min(m^{1-\PartFrac},n))$ time
if we maintain the data structure $\DS_{\PartFrac}$ with auxiliary counts
\Aux{vLV} and \Aux{uLv}.
\end{lemma}
\begin{proof}
The number of length-$3$ paths containing $\Edge{u}{v}$ as centerpiece is
$\DegreeEx{v}(u) \cdot \DegreeEx{u}(v)$ minus the number of $3$-cycles
containing $\Edge{u}{v}$, which can be obtained in worst-case
$\bigO(\min(m^{1-\PartFrac},n))$ time by \Lemma{count-triangle-edge}
if $\DS_{\PartFrac}$ is maintained with auxiliary count
\Aux{uLv}.

The number of length-$3$ paths that start or end with $\Edge{u}{v}$
and contain two further vertices $x, y$, \ie,
$u - v - x - y$ or $v - u - x - y$,
equals $\Aux{vLV[$v$]} - \DegreeEx{v}(u) + \Aux{vLV[$u$]} - \DegreeEx{u}(v)$ if
$x$ has low degree.
For the case where $x$ has high degree, we iterate over
all $h \in \High \setminus \Set{u, v}$.
If $h$ is adjacent to $u$, we add $\DegreeEx{v}(h)-1$
and if $h$ is adjacent to $v$, we add $\DegreeEx{u}(h)-1$.
This takes $\bigO(\Card{\High}) = \bigO(\min(m^{1-\PartFrac},n))$ time.
\end{proof}
\begin{corollary}\CorLabel{query-3path-edge}
Let $G = (V, E)$ be a dynamic graph, $\PartFrac \in \Range{0}{1}$,
and let $\Graphlet$ be the length-$3$ path~\PictoThreePath{}.
We can query
$\Count(G, \Graphlet, \Edge{u}{v})$
for an arbitrary edge $\Edge{u}{v} \in E$ in
worst-case
$\bigO(\min(m^{1-\PartFrac},n))$ time,
with
$\bigO(m^{\PartFrac})$
amortized update time
and $\bigO(\min(m^{1+\PartFrac}, n^2))$ space.
\end{corollary}
\begin{proof}
Follows directly from \Lemma{count-3path-edge}, \Lemma{maintain-vLV}, and \Lemma{maintain-uLv}.
\end{proof}
\begin{lemma}\LemLabel{count-3path}
Let $G = (V, E)$ be a dynamic graph and $\Graphlet$ the
length-$3$ path~\PictoThreePath{}.
We can maintain $\Count(G, \Graphlet)$ in amortized $\bigO(\sqrt{m})$ update time and
$\bigO(\min(m^{1.5}, n^2))$ space.
We can query $\Count(G, \Graphlet, e)$ for an arbitrary edge $e \in E$
in worst-case $\bigO(\sqrt{m})$ time.
\end{lemma}
\begin{proof}%
After an edge $\Edge{u}{v}$ was inserted or before an edge $\Edge{u}{v}$ is
removed, the number of length-$3$ paths containing $\Edge{u}{v}$ can be
obtained in $\bigO(\min(m^{1-\PartFrac},n))$ worst-case time by
\Lemma{count-3path-edge} if $\DS_{\PartFrac}$ is maintained with auxiliary
counts \Aux{vLV} and \Aux{uLv}.
By \Lemma{maintain-vLV} and \Lemma{maintain-uLv}, this can be done with
$\PartFrac = \frac{1}{2}$
in amortized $\bigO(\sqrt{m})$ time and $\bigO(\min(m^{1.5}, n^2))$ space.
The additional time for the query is $\bigO(m^{1-\PartFrac}) =
\bigO(\sqrt{m})$.
By \Lemma{count-3path-edge}, the worst-case time to query $\Count(G, \Graphlet,
e)$ for an arbitrary edge $e \in E$ then is $\bigO(\sqrt{m})$.
\end{proof}
\paragraph*{Counting Claws}
\begin{lemma}\LemLabel{count-claw-vertex}
Let $G = (V, E)$ be a dynamic graph, $\PartFrac \in \Range{0}{1}$.
We can query the number of claws containing
two arbitrary vertices $u$ and $v$, $u \neq v \in V$,
as non-central vertices in worst-case
$\bigO(m^{1-\PartFrac})$ time
if we maintain the data structure $\DS_{\PartFrac}$ with auxiliary count
\Aux{cLV}.
\end{lemma}
\begin{proof}
The number of claws where the central vertex has low degree
is given directly by \Aux{cLV[$u, v$]}.
To count all claws where the central vertex has high degree, we iterate over
all $h \in \High$, test whether $\Edge{u}{h}, \Edge{v}{h}\in E$,
and add up $\Degree(h)-2$.
The running time is hence $\bigO(\Card{\High}) = \bigO(m^{1-\PartFrac})$.
\end{proof}
\paragraph*{Counting Paws}
\begin{lemma}\LemLabel{count-paw-edge}
Let $G = (V, E)$ be a dynamic graph,
$\PartFrac \in \Range{0}{1}$, and
$\Graphlet$ be the paw~\PictoPaw{}.
We can query
$\Count(G, \Graphlet, \Edge{u}{v})$
for an arbitrary edge $\Edge{u}{v} \in E$ in
worst-case $\bigO(\min(m^{\max(1-\PartFrac,\PartFrac)},n))$ time
if we maintain the data structure $\DS_{\PartFrac}$ with auxiliary counts
\Aux{uLv},
\Aux{t}, and
\Aux{cLV}.
\end{lemma}
\begin{proof}
Edge
$\Edge{u}{v}$ may be either (1) the arm, (2) one of the two triangle edges
connected to the central vertex, or (3) the third triangle edge.
Let $t_{\Edge{u}{v}}$ be the number of $3$-cycles containing $\Edge{u}{v}$,
which can be obtained in $\bigO(\min(m^{1-\PartFrac},n))$ worst-case time by
\Lemma{count-triangle-edge} if we maintain $\DS_{\PartFrac}$ with auxiliary
count \Aux{uLv}.
Let $t_u$ ($t_v$) be the number of triangles containing $u$ ($v$),
which can be obtained
by \Lemma{query-triangle-vertex}
in
$\bigO(\min(m^{2\PartFrac},n^2))$ worst-case time by
\Lemma{query-triangle-vertex}
if we maintain the data structure $\DS_{\PartFrac}$ with auxiliary count
\Aux{t}. %

In case (1), the number of paws equals the number of $3$-cycles containing
either $v$ or $u$, but not both, \ie, $t_u + t_v - 2t_{\Edge{u}{v}}$.
In case (2), the number of paws equals $t_{\Edge{u}{v}}$ multiplied by
$\DegreeEx{v}(u) - 1 + \DegreeEx{u}(v)-1$, and in the case (3), it is
the number of claws with low-degree central vertex, \Aux{cLV[$u, v$]}, plus those
with a high-degree central vertex.
This number can be obtained in $\bigO(\Card{\High}) =
\bigO(\min(m^{1-\PartFrac},n))$ time by computing the sum of $\Degree(h)-2$ for
all $h \in \High$ with $\Edge{u}{h}, \Edge{v}{h} \in E$.
\end{proof}
\begin{lemma}\LemLabel{count-paw}
Let $G = (V, E)$ be a dynamic graph and $\Graphlet$ the paw~\PictoPaw{}.
We can maintain $\Count(G, \Graphlet)$ in amortized $\bigO(m^{2/3})$ update time and
$\bigO(n^2)$ space
and query $\Count(G, \Graphlet, e)$ for an arbitrary edge $e \in E$
in worst-case $\bigO(m^{2/3})$ time.
Alternatively,
we can maintain $\Count(G, \Graphlet)$ in amortized $\bigO(m)$ update time and
$\bigO(n^2)$ space
and query $\Count(G, \Graphlet, e)$ for an arbitrary edge $e \in E$
in worst-case $\bigO(\sqrt{m})$ time.
\end{lemma}
\begin{proof}%
After an edge $\Edge{u}{v}$ was inserted or before an edge $\Edge{u}{v}$ is removed,
the number of paws containing it can be obtained in 
$\bigO(\min(m^{\max(1-\PartFrac,2\PartFrac)},n^2))$ worst-case time
by \Lemma{count-paw-edge} if $\DS_{\PartFrac}$ with auxiliary counts
\Aux{uLv},
\Aux{t}, and
\Aux{cLV}
is maintained.
By \Lemma{maintain-uLv}, \Lemma{maintain-t}, and \Lemma{maintain-cLV},
this can be done in amortized
$\bigO(m^{\PartFrac} + m^{\max(1-\PartFrac,\PartFrac)} + m^{2\PartFrac})
 = \bigO(m^{\max(1-\PartFrac,2\PartFrac)})$ time and $\bigO(n^2)$ space.
Together with the cost for the query,
this yields an amortized update time of $\bigO(m^{\max(1-\PartFrac,2\PartFrac)}) = \bigO(m^{2/3})$
with $\PartFrac = \frac{1}{3}$.
By \Lemma{count-paw-edge}, the worst-case time to query $\Count(G, \Graphlet,
e)$ for an arbitrary edge $e \in E$ then is $\bigO(m^{2/3})$.
\end{proof}

\paragraph*{Counting $4$-Cycles}
\begin{lemma}\LemLabel{count-4cycle-edge}
Let $G = (V, E)$ be a dynamic graph, $\PartFrac \in \Range{0}{1}$, and
$\Graphlet$ be the $4$-cycle~\PictoFourCycle{}.
We can query
$\Count(G, \Graphlet, \Edge{u}{v})$
for an arbitrary edge $\Edge{u}{v} \in E$ in
worst-case
$\bigO(\min(m^{\max(1-\PartFrac,2\PartFrac)},n^2))$ time
if we maintain the data structure $\DS_{\PartFrac}$ with auxiliary counts
\Aux{uLLv},
\Aux{uLv}, and
\Aux{uHv}.
\end{lemma}
\begin{proof}
The $4$-cycles containing an edge $\Edge{u}{v}$
can be distinguished by the number of low-degree vertices among the other two
vertices $a, b$ forming the cycle $\Sequence{\Edge{u}{v}, \Edge{v}{a},
\Edge{a}{b}, \Edge{b}{u}}$.

The number of cycles where $a, b \in \Low$ is given by \Aux{uLLv[$u, v$]}.

For the number of cycles with $a \in \Low, b \in \High$ (analogously vice-versa with
swapped rules for $u$ and $v$), we use length-$2$ paths as follows:
We iterate over all $h \in \High \setminus \Set{v}$ such that $\Edge{u}{h}\in
E$ and obtain \Aux{uLv[$v, h$]} in total $\bigO(\Card{\High}) =
\bigO(\min(m^{1-\PartFrac},n))$ time.
Each such vertex $h$ is a possible choice for $b$.
If $u \in \High$,
the number of cycles equals the sum of \Aux{uLv[$v, h$]} over all such $h$.
Otherwise, \Aux{uLv[$v, h$]} also counts the length-$2$ path
$\Sequence{\Edge{h}{u}, \Edge{u}{v}}$, so we instead sum up $\Aux{uLv[$v, h$]}
- 1$ for all such $h$.

For all cycles with $a, b \in \High$, we proceed as follows:
If $u, v \in \Low$, we iterate over all pairs of neighbors $x \in \NeighEx{v}(u) \cap \High$
and $y \in \NeighEx{u}(v) \cap \High$
in $\bigO(\Degree(u) \cdot \Degree(v)) = \bigO(\min(m^{2\PartFrac},n^2))$ time.
The number of $4$-cycles equals the number of pairs $x, y$ with $\Edge{x}{y} \in E$.
Otherwise, assume \wilog{} that $u \in \High$.
If $v \in \Low$, %
the number of $4$-cycles is the sum of \Aux{uHv[$u, h$]} over all $h \in \High$
with $\Edge{h}{v}\in E$,
and the sum of $(\Aux{uHv[$u, h$]} - 1)$ otherwise.
This can be done in $\bigO(\Card{H}) = \bigO(\min(m^{1-\PartFrac},n))$ time.
\end{proof}
\begin{lemma}\LemLabel{count-4cycle}
Let $G = (V, E)$ be a dynamic graph and $\Graphlet$ the
$4$-cycle~\PictoFourCycle{}.
We can maintain $\Count(G, \Graphlet)$ in amortized $\bigO(m^{2/3})$ update time and
$\bigO(n^2)$ space.
We can query $\Count(G, \Graphlet, e)$ for an arbitrary edge $e \in E$
in worst-case $\bigO(m^{2/3})$ time.
\end{lemma}
\begin{proof}%
After an edge $\Edge{u}{v}$ was inserted or before an edge $\Edge{u}{v}$ is removed,
the number of $4$-cycles containing it can be obtained in
$\bigO(\min(m^{\max(1-\PartFrac,2\PartFrac)},n^2))$
worst-case time
by \Lemma{count-4cycle-edge} if $\DS_{\PartFrac}$ with auxiliary counts
\Aux{uLv},
\Aux{uLLv}, and
\Aux{uHv}
is maintained.
By \Lemma{maintain-uLv}, \Lemma{maintain-uLLv}, and
\Lemma{maintain-uHv}, this can be done in amortized
$\bigO(m^{\PartFrac} + m^{2\PartFrac} + m^{\max(1-\PartFrac,\PartFrac)})
 = \bigO(m^{\max(1-\PartFrac,2\PartFrac)})$ time and $\bigO(n^2)$ space.
Together with the cost for the query, this yields a total amortized update time
of $\bigO(m^{\max(1-\PartFrac,2\PartFrac)}) = \bigO(m^{2/3})$ with $\PartFrac =
\frac{1}{3}$.
By \Lemma{count-4cycle-edge}, the worst-case time to query $\Count(G, \Graphlet,
e)$ for an arbitrary edge $e \in E$ then is $\bigO(m^{2/3})$.
\end{proof}

\subsection{Counting Non-Induced $s$-Subgraphs}\SectLabel{app-s-subgraph}
\begin{lemma}\LemLabel{count-s-triangle}
Let $G = (V, E)$ be a dynamic graph, $\Graphlet$ be a $3$-cycle~\PictoThreeCycle{} and $s \in V$.
We can maintain $\Count(G, \Graphlet, s)$ in amortized $\bigO(\sqrt{m})$ update time and
$\bigO(n)$ space.
\end{lemma}
\begin{proof}%
If $\Edge{u}{v}$ is inserted or deleted and $u \neq s \neq v$,
the count changes by one if $\Edge{u}{s}, \Edge{v}{s} \in E$.
Otherwise, if $\Edge{s}{x}$ is inserted or deleted,
the count changes by the number of triangles containing $\Edge{s}{x}$,
which can be obtained in $\bigO(m^{1-\PartFrac})$ time
by \Lemma{count-triangle-edge} if we maintain $\DS_{\PartFrac}$ with
\Aux{uLv}.

By \Lemma{maintain-uLv}, $\DS_{\PartFrac}$ with
\Aux{uLv} can be maintained in amortized
$\bigO(m^{\PartFrac})$ update time and $\bigO(n^2)$ space.
In total, we can maintain $3$-cycles containing $s$ in
amortized update time
$\bigO(m^{\max(1-\PartFrac,\PartFrac)})
= \bigO(\sqrt{m})$ with $\PartFrac = \frac{1}{2}$
and $\bigO(n^2)$ space.
The space can be reduced to $\bigO(n)$ by restricting \Aux{uLv} such
that it only maps single vertices and uses $s$ as the second.
\end{proof}

\begin{lemma}\LemLabel{count-s-3path}
Let $G = (V, E)$ be a dynamic graph, $\Graphlet$ be a length-$3$
path~\PictoThreePath{}, and $s \in V$.
We can maintain $\Count(G, \Graphlet, s)$ in amortized $\bigO(\sqrt{m})$ update time and
$\bigO(n)$ space.
\end{lemma}
\begin{proof}
If an edge $\Edge{s}{x}$ is inserted or deleted, we count
the number of affected paths
in worst-case $\bigO(\min(m^{1-\PartFrac},n))$ time
by \Lemma{count-3path-edge}
if we maintain the data structure $\DS_{\PartFrac}$ with auxiliary counts
\Aux{vLV} and \Aux{uLv}.

If an edge $\Edge{u}{v}$ is inserted or deleted and $u \neq s \neq v$,
the count changes by $\Degree(s)-1 + \DegreeEx{u}(v)$
if $\Edge{s}{u} \in E$ and $\Edge{s}{v} \not\in E$.
Vice versa, if $\Edge{s}{v} \in E$ and $\Edge{s}{u} \not\in E$,
it changes by $\Degree(s)-1 + \DegreeEx{v}(u)$.
If both $\Edge{s}{u}, \Edge{s}{v} \in E$, the count changes by
$2(\Degree(s)-2) + \DegreeEx{u}(v)-1 + \DegreeEx{v}(u)-1$.
This accounts for paths where $s$ is adjacent to either $u$ or $v$.
For paths via an intermediate low-degree vertex between $s$ and $u$, the count
changes by \Aux{uLv[$s, u$]}, subtracting $1$ if $\Edge{s}{v}\in E$,
$\Edge{u}{v}$ was added, and $v \in \Low$.
Symmetrically, the count changes by \Aux{uLv[$s, v$]}, subtracting $1$ if
$\Edge{s}{u}\in E$, $\Edge{u}{v}$ was added, and $u \in \Low$.
For paths with an intermediate high-degree vertex between $s$ and $u$ or $v$,
we count all $h \in \High \setminus\Set{v,s}$ with $\Edge{h}{s}, \Edge{h}{u} \in E$
as well as all $h' \in \High \setminus\Set{u,s}$ with $\Edge{h'}{s}, \Edge{h'}{g} \in E$.

By \Lemma{maintain-vLV} and \Lemma{maintain-uLv}, we can maintain
$\DS_{\PartFrac}$ with \Aux{vLV} and \Aux{uLv} in $\bigO(m^{\PartFrac})$
amortized time and $\bigO(n^2)$ space.
Iterating over all high-degree vertices can be done in $\bigO(\Card{\High}) =
\bigO(m^{1-\PartFrac})$ time, thus yielding an amortized update time of
$\bigO(m^{\max(\PartFrac,1-\PartFrac)}) = \bigO(\sqrt{m})$ time for $\PartFrac
= \frac{1}{2}$.
The space complexity can be reduced to $\bigO(n)$ by
only maintaining \Aux{uLv} for those vertex pairs where one vertex is $s$.
\end{proof}

\begin{lemma}\LemLabel{count-s-claw}
Let $G = (V, E)$ be a dynamic graph, $\Graphlet$ be the claw~\PictoClaw{}, and $s \in V$.
We can maintain $\Count(G, \Graphlet, s)$ in worst-case $\bigO(1)$ update time and
constant space.
\end{lemma}
\begin{proof}
Only the insertion or deletion of an edge $\Edge{s}{x}$ can change the number
of claws containing $s$ as the central vertex.
In this case, we update the count by
$\frac{1}{2}\DegreeEx{x}(s)\cdot(\DegreeEx{x}(s)-1)$ if $\DegreeEx{x}(s) \geq 2$.
Symmetrically, for the number of claws where $s$ is one of the non-central
vertices and $x$ is the center, we update the count by
$\frac{1}{2}\DegreeEx{s}(x)\cdot(\DegreeEx{s}(x)-1)$ if $\DegreeEx{s}(x) \geq 2$.

In case that an edge $\Edge{u}{v}$ is inserted or deleted and $u \neq s \neq v$,
the count changes by $\DegreeEx{v}(u)-1$ if $\DegreeEx{v}(u) \geq 1$ and
$\Edge{s}{u} \in E$,
and by $\DegreeEx{u}(v)-1$ if $\DegreeEx{u}(v) \geq 1$ and
$\Edge{s}{v} \in E$.

All steps can be done in constant time and space.
\end{proof}

\begin{lemma}\LemLabel{count-s-paw}
Let $G = (V, E)$ be a dynamic graph, $\Graphlet$ be the paw \PictoPaw{}, and $s \in V$.
We can maintain $\Count(G, \Graphlet, s)$ in amortized
$\bigO(m^{2/3})$ time and $\bigO(n^2)$ space.
\end{lemma}
\begin{proof}
There are three possibilities for a paw to contain $s$:
(1) at the center,
(2) at the other end of the arm, or
(3) as one of the two non-central triangle vertices.
Let $t_v$ be the number of triangles containing a vertex $v$ after the
insertion or before the deletion of an edge, respectively.

If an edge $\Edge{s}{x}$ is inserted or deleted,
then by \Lemma{count-paw-edge},
we obtain the number of paws containing $\Edge{s}{x}$
in
worst-case $\bigO(\min(m^{\max(1-\PartFrac,2\PartFrac)},n^2))$ time
if we maintain the data structure $\DS_{\PartFrac}$ with auxiliary counts
\Aux{uLv},
\Aux{t}, and
\Aux{cLV}.

If an edge $\Edge{u}{v}$ is inserted or deleted and $u \neq s \neq v$,
we distinguish two cases.
If $\Edge{s}{u}\in E$ ($\Edge{s}{v}\in E$), the count changes by the difference
in the number of paws containing $\Edge{s}{u}$ ($\Edge{s}{v}$) before and after
the update.
The difference can be obtained in at most four queries using again
\Lemma{count-paw-edge}.
If both $\Edge{s}{u}, \Edge{s}{v}\in E$, we subtract
$\deg(s) - 2  + \deg(u) -2 + \deg(v)-2$
from the sum of both differences to avoid double-counting paws that contain the
triangle $s, u, v$.
By this, we cover possibility (1), (3), and partially also (2).
The remaining situation is that $s$ is the non-central vertex incident
to the arm and $\Edge{u}{v}$ is not incident to the central vertex.
The number of paws in this case equals the number of claws containing $u$,
$v$, and $s$ with a low-degree central vertex, which is given by \Aux{cL[$u, v,
s$]}, and the number of claws with a high-degree central vertex, which we
obtain by counting the number of vertices $h \in \High$ with $\Edge{h}{s},
\Edge{h}{u}, \Edge{h}{v}\in E$ in $\bigO(\Card{\High}) =
\bigO(m^{1-\PartFrac})$ time.

By \Lemma{maintain-uLv}, \Lemma{maintain-t}, and \Lemma{maintain-cLV},
we can maintain
\Aux{uLv},
\Aux{t}, and
\Aux{cLV}
in amortized
$\bigO(m^{\PartFrac} + m^{\max(1-\PartFrac,\PartFrac)} + m^{2\PartFrac})
 = \bigO(m^{\max(1-\PartFrac,2\PartFrac)})$ time and $\bigO(n^2)$ space.
\Aux{cL} can be maintained in $\bigO(m^{2\PartFrac})$ time and
$\bigO(n\min(n^2,m^{3\PartFrac}))$ space by \Lemma{maintain-cL}.
With additional
$\bigO(\min(m^{\max(1-\PartFrac,2\PartFrac)},n^2) + m^{1-\PartFrac})$
time to update the count, we obtain a total amortized update time
of $\bigO(m^\frac{2}{3})$ with $\PartFrac = \frac{1}{3}$.
The total space complexity can be reduced to $\bigO(n^2)$ by not
maintaining \Aux{cL} for arbitrary triples of vertices, but
only for pairs and fixing $s$ as third vertex.
\end{proof}
\begin{lemma}\LemLabel{count-s-4cycle}
Let $G = (V, E)$ be a dynamic graph, $\Graphlet$ be the $4$-cycle \PictoFourCycle{} and $s \in V$.
We can maintain $\Count(G, \Graphlet, s)$ in amortized
$\bigO(m^{2/3})$ time and $\bigO(n^2)$ space.
\end{lemma}
\begin{proof}
If an edge $\Edge{s}{x}$ is inserted or deleted, the count changes by
the number of $4$-cycles containing $\Edge{s}{x}$, which can be
obtained in worst-case
$\bigO(\min(m^{\max(1-\PartFrac,2\PartFrac)},n^2))$ time
by \Lemma{count-4cycle-edge}
if $\DS_{\PartFrac}$ with auxiliary counts
\Aux{uLLv},
\Aux{uLv}, and
\Aux{uHv} is maintained.

If an edge $\Edge{u}{v}$ is inserted or deleted and $u \neq s \neq v$,
and $\Edge{s}{u} \in E$,
the count changes by the number of length-2 paths between $s$ and $v$
besides $\Sequence{\Edge{s}{u}, \Edge{u}{v}}$,
which the sum of $\Aux{uLv[$s, v$]}$
and the number of high-degree vertices $h \in \High$ with $\Edge{s}{h},
\Edge{h}{v}$ minus 1.
Analogously, if $\Edge{s}{v} \in E$, the count changes by the number of
length-2 paths between $s$ and $u$ besides $\Sequence{\Edge{s}{v},
\Edge{v}{u}}$, which is $\Aux{uLv[$s, u$]}$ plus the number of high-degree
vertices $h \in \High$ with $\Edge{s}{h}, \Edge{h}{u}$
minus 1.
This can be done in $\bigO(\Card{\High}) = \bigO(m^{1-\PartFrac})$ time.

By
\Lemma{maintain-uLv},
\Lemma{maintain-uLLv}, and
\Lemma{maintain-uHv},
we can maintain $\DS_{\PartFrac}$
with an amortized update time of
$\bigO(m^{\max(2\PartFrac,1-\PartFrac)})$ and $\bigO(n^2)$ space.
Together with the additional time for updating the count, we arrive at an
amortized update time of $\bigO(m^{\max(2\PartFrac,1-\PartFrac)}) =
\bigO(m^{2/3})$ with $\PartFrac = \frac{1}{3}$ and a space complexity of
$\bigO(n^2)$.
\end{proof}

\begin{lemma}\LemLabel{count-s-diamond}
Let $G = (V, E)$ be a dynamic graph, $\Graphlet$ be the diamond \PictoDiamond{} and $s \in V$.
We can maintain $\Count(G, \Graphlet, s)$ in amortized
$\bigO(m^{2/3})$ update time and $\bigO(n^2)$ space.
\end{lemma}
\begin{proof}
If an edge $\Edge{s}{x}$ is inserted or deleted, the change in the number of
diamonds can be obtained in
$\bigO(\min(m^{\max(1-\PartFrac,2\PartFrac)},n^2))$ worst-case time
by \Lemma{count-diamond-edge}
if we maintain the data structure $\DS_{\PartFrac}$ with auxiliary counts
\Aux{uLv}, \Aux{pLL}, \Aux{uHv}, and \Aux{cL}.

If an edge $\Edge{u}{v}$ is inserted or deleted and $u \neq s \neq v$,
we distinguish three cases:
(1) $\Edge{u}{v}$ is the chord,
(2) $s$ is incident to the chord,
(3) $u$ or $v$ is incident to the chord, but not $s$ and $\Edge{u}{s}$ or $\Edge{v}{s} \in E$.
For case (1), if $\Edge{u}{s}, \Edge{v}{s} \in E$, the count changes by the
number of triangles containing $\Edge{u}{v}$ minus $1$ to correct for the
triangle containing $s$.
For case (2), if $\Edge{u}{s}, \Edge{v}{s} \in E$, the count changes by the
number of triangles containing $\Edge{s}{u}$ minus $1$
plus
the number of triangles containing $\Edge{s}{v}$ minus $1$,
correcting with minus $1$ again for triangles containing $v$ and $u$,
respectively.
For case (3),
the count changes by the number of claws \Aux{cL[$s, u, v$]} with
a low-degree fourth vertex plus the number of high-degree vertices $h \in \High$
such that $\Edge{h}{s}, \Edge{h}{u}, \Edge{h}{v} \in E$
if either $\Edge{s}{u}$ or $\Edge{s}{v} \in E$, and by twice this amount
if both $\Edge{s}{u}, \Edge{s}{v} \in E$.

The number of triangles containing an edge $\Edge{u}{v}$, $\Edge{s}{u}$,
or $\Edge{s}{v}$
can be queried in $\bigO(m^{1-\PartFrac})$ worst-case time if $\DS_{\PartFrac}$
with \Aux{uLv} is maintained by \Lemma{count-triangle-edge}.
We hence need to maintain $\DS_{\PartFrac}$ with
auxiliary counts
\Aux{uLv}, \Aux{pLL}, \Aux{uHv}, and \Aux{cL}.
By
\Lemma{maintain-uLv},
\Lemma{maintain-pLL},
\Lemma{maintain-uHv}, and
\Lemma{maintain-cL}
and together with the query times to update the count,
we arrive at an
amortized total update time of $\bigO(m^{\max(2\PartFrac,1-\PartFrac)}) =
\bigO(m^{2/3})$ with $\PartFrac =
\frac{1}{3}$. %
The space complexity can be reduced to $\bigO(n^2)$ by maintaining \Aux{cL}
only for those triplets that contain $s$.
\end{proof}

\begin{lemma}\LemLabel{count-s-4clique}
Let $G = (V, E)$ be a dynamic graph, $\Graphlet$ be the $4$-clique
\PictoClique{}, and $s \in V$.
We can maintain $\Count(G, \Graphlet, s)$ in amortized
$\bigO(m)$ update time and constant space.
\end{lemma}
\begin{proof}
If an edge $\Edge{s}{x}$ is inserted or deleted,
we obtain the change in the number of $4$-cliques by
querying  $\Count(G, \Graphlet, \Edge{s}{x})$,
which can be done in $\bigO(m)$ worst-case time
and constant space by \Theorem{count-any}.
If an edge $\Edge{u}{v}$ is inserted or deleted, $u \neq s \neq v$,
and $\Edge{u}{s}, \Edge{v}{s}\in E$,
we iterate over all $y \in V$ and test whether $\{s, u, v, y\}$ form
a $4$-clique.
This takes $\bigO(n)$ time.
In total, we arrive at a worst-case update time of $\bigO(m)$ and constant space.
\end{proof}
\section{Omitted Proofs from \Section{lowerbounds}}\SectLabel{lb-proofs}
\begin{observation}\ObLabel{cycle-side-long}
Every path in $G_{M,g,h}$ that contains two distinct vertices
$l_i \neq l_{i'} \in L$
(resp.~$r_j \neq r_{j'} \in R$) has length at least $2g+2$ (resp.~$2h+2$).
\end{observation}

Each case in \Lemma{subgraph-det-cnt} corresponds to one of the following lemmas:
\paragraph*{($s$-)$k$-Odd Cycle Detection}
\begin{lemma}[($s$-)$k$-Odd Cycle Detection]\LemLabel{k-odd-cycle}
Given a partially dynamic algorithm $\Alg$ for ($s$-)$k$-cycle detection with $k
\geq 3$ and $k$ odd, one can solve \oneuMv{} with parameters $n_1$ and $n_2$ by running the
preprocessing step of $\Alg$ on a graph with $\bigO(m + \sqrt{m}\cdot k)$ edges
and $\Theta(\sqrt{m}\cdot k)$ vertices, and then making $\bigO(\sqrt{m})$
insertions (or $\bigO(\sqrt{m})$ deletions) and $1$ query, where $m$ is such
that $n_1=n_2=\sqrt{m}$.
\end{lemma}
\begin{proof}%
We only prove the decremental case, the incremental case has a symmetric proof.
Consider a \oneuMv{} problem with
$n_1 = n_2 = \sqrt m$.
Given $M$, we construct the $(k-1)$-partite graph $G$ from
$G_{M,g,h}$ with $g = h = \frac{1}{2}(k-3)$
by adding to it a vertex $s$
and edges $\Edge{s}{l_i}$, $\Edge{r_j}{s}$ for all $r_j \in R$, $l_i\in L$.
Thus, the total number of edges is at most
$n_1n_2 + \frac{1}{2}(k-1)(n_1 + n_2)=\bigO(m + k\cdot\sqrt{m})$.
Once $u$ and $v$ arrive, we delete $\Edge{s}{l_i}$ iff $u_i=0$ and
delete $\Edge{r_j}{s}$ iff $v_j=0$.
See \Figure{lb-examples} for an example with $k=5$.

We have $u^\top Mv= 1$ iff there is a cycle of length $k$ in $G$ iff
there is a cycle of length $k$ incident to $s$:
By \Observation{cycle-bip-even},
every cycle that does not contain $s$ is even.
By \Observation{cycle-side-long}, every cycle that contains $s$ but
does not contain both a vertex in $L$ and $R$ has length at least $k+1$.
In total, we need to do $n_1+n_2=\bigO(\sqrt{m})$ updates and $1$ query.
\end{proof}

\paragraph*{($s$-)Paw Detection}
\begin{lemma}[($s$-)Paw Detection]\LemLabel{lb-paw}
Given a partially dynamic algorithm $\Alg$ for ($s$-)paw detection, one can
solve \oneuMv{} with parameters $n_1$ and $n_2$ by running the preprocessing
step of $\Alg$ on a graph with $\bigO(m + \sqrt{m}\cdot k)$ edges and
$\Theta(\sqrt{m}\cdot k)$ vertices, and then making $\bigO(\sqrt{m})$
insertions (or $\bigO(\sqrt{m})$ deletions) and $1$ query, where $m$ is such
that $n_1=n_2=\sqrt{m}$.
\end{lemma}
\begin{proof}%
Construct $G$ as in the proof of \Lemma{k-odd-cycle} for $k = 3$ and add an
additional vertex $t$ and edge $\Edge{s}{t}$.
Update $G$ as in the proof of \Lemma{k-odd-cycle}.
Then, $u^\top Mv= 1$ iff there is a paw in $G$ iff there is a paw incident to
$s$ and
we need to do $n_1+n_2=\bigO(\sqrt{m})$ updates and $1$ query.
\end{proof}
\begin{lemma}[($s$-)$4$-Clique Detection]\LemLabel{lb-clique}
Given a partially dynamic algorithm $\Alg$ for ($s$-)$4$-clique detection,
one can solve \oneuMv{} with parameters $n_1$ and $n_2$ by running the
preprocessing step of $\Alg$ on a graph with $\bigO(m)$ edges
and $\Theta(\sqrt{m})$ vertices, and then making $\bigO(\sqrt{m})$
insertions (or $\bigO(\sqrt{m})$ deletions) and $1$ query, where $m$ is such
that $n_1=n_2=\sqrt{m}$.
\end{lemma}
\begin{proof}[Proof of \Lemma{lb-clique}]
We only prove the decremental case.
Consider a \oneuMv{} problem with
$n_1 = n_2 = \sqrt m$.
Given $M$, we construct the tripartite graph $G$ from $G_{M,1,0}$
by adding to it a vertex $s$
and edges
$\Edge{s}{l_i}$,
$\Edge{s}{l^{(1)}_i}$,
$\Edge{r_j}{s}$,
$\Edge{l_i}{r_j}$
for all $r_j \in R$, $l_i\in L$.
Thus, the total number of edges is at most
$2n_1n_2 + 3n_1 + n_2=\bigO(m)$.
Once $u$ and $v$ arrive, we delete
$\Edge{s}{l_i}$ and $\Edge{s}{l^{(1)}_i}$ iff $u_i=0$ and
delete $\Edge{r_j}{s}$ iff $v_j=0$.
See \Figure{lb-examples} for an example.

Consider the case that $G$ contains a $4$-clique (\ie, $K_4$). %
Note that each edge of the clique is shared by exactly two triangles, whose
vertices together induce $K_4$.
In consequence of \Observation{cycle-bip-even}, every triangle must be
incident to $s$ or contain an edge $\Edge{l_i}{r_j}$, $l_i \in L, r_j \in R$
for some $1 \leq i \leq n_1$, $1 \leq j \leq n_2$.
As $s$ can be only one vertex of $K_4$ and the three remaining vertices also
form a triangle in $G$, $K_4$ must contain an edge $\Edge{l_i}{r_j}, l_i \in L,
r_j \in R$ for some $1 \leq i \leq n_1$, $1 \leq j \leq n_2$.
By construction, each vertex $l_i \in L$ has an edge to $l^{(1)}_i$ and
possibly $s$, and apart from that has only neighbors in $R$.
Thus, $\Edge{l_i}{r_j}$ is always part of a triangle $\{l_i,
l^{(1)}_i, r_j\}$, and possibly a second triangle $\{s, l_i, r_j\}$.
However, by construction, no triangle $\{l_i, r_j, r_{j'}\}$ with $j \neq j'$
can exist, as $G_{M,1,0}$ is tripartite and no edge between two vertices in $R$
has been added.
Thus, if $\Edge{l_i}{r_j}$ is part of $K_4$,
there are triangles $\{l_i,l_i^{(1)}, r_j\}$ and $\{s, l_i, r_j\}$.
By construction, $G$ then also contains $\Edge{s}{l^{(1)}_i}$, which implies
that $K_4$ is induced by the vertices $\{s, l_i, l^{(1)}_i, r_j\}$.

Hence, we have $u^\top Mv= 1$
iff there is a $4$-clique in $G$
iff $s$ is a vertex of a $4$-clique in $G$.
In total, we need to do $2n_1+n_2=\bigO(\sqrt{m})$ updates and $1$ query.
\end{proof}
\begin{lemma}[($s$-)Length-$3$ Path Counting]\LemLabel{three-path-counting}
Given a partially dynamic algorithm $\Alg$ for counting the number of ($s$-)length-$3$ paths,
one can solve \oneuMv{} with parameters $n_1$ and $n_2$ by
running the preprocessing step of $\Alg$ on a graph with $\bigO(m +
\sqrt{m}\cdot k)$ edges and $\Theta(\sqrt{m}\cdot k)$ vertices, and then making
$\bigO(\sqrt{m})$ insertions (or $\bigO(\sqrt{m})$ deletions) and $1$ query,
where $m$ is such that $n_1=n_2=\sqrt{m}$.
\end{lemma}
\begin{proof}
We only prove the incremental case, the decremental case is symmetric.
Consider a \oneuMv{} problem with
$n_1 = n_2 = \sqrt m$.
Given $M$, we construct the tripartite graph $G$ from $G_{M,g,h}$ with $g
= h = 0$ by
adding two isolated vertices $s, t$.
We also count for each vertex $v \in R \cup L$ the number of paths
of length $2$ starting at $v$
and the total number of length-$3$ paths in $G$.
Once $u$ and $v$ arrive, we add $\Edge{s}{l_i}$ iff $u_i=1$ and
$\Edge{r_j}{t}$ iff $v_j=1$.
Furthermore, we count the number $k_u$ of $1$-entries in $u$
and the number $k_v$ of $1$-entries in $v$.
If $u_i=1$, we increment the length-$3$ path counter by
the number of paths of length $2$ starting at $l_i$
plus $(k_u-1)\cdot\Degree(l_i)$.
Likewise, if $v_j=1$, we increment the length-$3$ path counter by
the number of paths of length $2$ starting at $r_j$
plus $(k_v-1)\cdot\Degree(r_j)$.
We have $u^\top Mv= 1$ iff the number of $3$-paths in $G$
is greater than the length-$3$ path counter.
In total, we need to do $n_1+n_2=\bigO(\sqrt{m})$ updates and $1$ query.
\end{proof}

For lower bounds on detecting and counting cycles of even length,
we modify $G_{M,g,h}$ by duplicating the left-hand side vertices
as well as the edges depending on $M$
(see \Figure{lb-construct-H} for an example):
\begin{definition}[$H_{M,g,h}$]\DefLabel{HMgh}
Given a matrix $M \in {\{0,1\}}^{n_1\times n_2}$ and two integers
$g, h \geq 1$, we denote by
$H_{M,g,h} =
(T^{(0)} \cup \dots \cup T^{(g)}
\cup B^{(0)} \cup \dots \cup B^{(g)}
\cup R^{(0)} \cup \dots \cup R^{(h)}
, E_L \cup E_R \cup E_{M_T} \cup E_{M_B})$ the
$(g+h+2)$-partite graph where
\setlength{\abovedisplayskip}{1pt}
\setlength{\belowdisplayskip}{1pt}
\begin{align*}
T^{(p)} &=\left\{t^{(p)}_1,\dots, t^{(p)}_{n_1}\right\}, \quad %
B^{(p)} =\left\{b^{(p)}_1,\dots, b^{(p)}_{n_1}\right\}, \quad 0 \leq p \leq g, \\
E_L &= \left\{\Edge{t^{(p)}_{i}}{t^{(p+1)}_{i}}, \Edge{b^{(p)}_{i}}{b^{(p+1)}_{i}} \mid 1 \leq i \leq n_1 \wedge 1 \leq p < g \right\} \\
R^{(q)} &=\left\{r^{(q)}_1,\dots, r^{(q)}_{n_2}\right\}, \quad 0 \leq q \leq h; \qquad
E_R = \{(r^{(q)}_{j}, r^{(q+1)}_{j}) \mid 1 \leq j \leq n_2 \wedge 1 \leq q < h \} \\
E_{M_T} &= \left\{\Edge{t^{(0)}_{i}}{r^{(0)}_j}, \Edge{t^{(g)}_{i}}{r^{(h)}_{j}} \mid M_{ij} = 1 \wedge i \leq j \right\} \\
E_{M_B} &= \left\{(b^{(0)}_{i}, r^{(0)}_j), (b^{(h)}_{i}, r^{(h)}_{j}) \mid M_{ij} = 1  \wedge i > j \right\}.
\end{align*}
$H_{M,g,h}$ has $2\cdot (g+1) \cdot n_1 + (h+1) \cdot n_2$ vertices and at most
$2n_1n_2 + 2\cdot g\cdot n_1 + h\cdot n_2$ edges.
All vertices $x$ in $T^{(p)}$, $B^{(p)}$ for $1 \leq p < g$ and in $R^{(q)}$ for $1 \leq q < h$ have $\Degree(x) = 2$.
\end{definition}
\begin{observation}\ObLabel{h-tree}
$H_{M,g,h}$ is acyclic (a forest).
\end{observation}
\begin{observation}\ObLabel{h-cycle-side-long}
Every path in $H_{M,g,h}$ that contains two distinct vertices
$x_i \neq x_{j} \in T^{(1)} \cup B^{(1)}$
($r_i \neq r_{j} \in R^{(1)}$) has length at least
$2g$ ($2h$).
\end{observation}
\begin{lemma}[$k$-Cycle Detection]\LemLabel{lb-cycle}
Given a partially dynamic algorithm $\Alg$ for $k$-cycle detection with $k
\geq 4$, one can solve \oneuMv{} with parameters $n_1$ and $n_2$ by running the
preprocessing step of $\Alg$ on a graph with $\bigO(m + \sqrt{m}\cdot k)$ edges
and $\Theta(\sqrt{m}\cdot k)$ vertices, and then making $\bigO(\sqrt{m})$
insertions (or $\bigO(\sqrt{m})$ deletions) and $1$ query, where $m$ is such
that $n_1=n_2=\sqrt{m}$.
\end{lemma}
\begin{proof}%
Consider a \oneuMv{} problem with
$n_1 = n_2 = \sqrt m$.
Given $M$, we construct the $k$-partite graph $H$ from
$H_{M,g,h}$ with $g = \lceil\frac{k}{2}\rceil - 1$, $h = \lfloor\frac{k}{2}\rfloor - 1$
by adding edges
$\Edge{t^{(0)}_i}{t^{(1)}_i}$,
$\Edge{b^{(0)}_{i'}}{b^{(1)}_{i'}}$,
$\Edge{r^{(0)}_j}{r^{(1)}_j}$,
for all $t^{(0)}_i \in T, b^{(0)}_{i'} \in B, r_j \in R^{(0)}$.
Thus, the total number of edges is at most
$2n_1n_2 + 2(\lceil\frac{k}{2}\rceil-1)\cdot n_1 + (\lfloor\frac{k}{2}\rfloor - 1)\cdot n_2=\bigO(m + k\cdot\sqrt{m})$.
Once $u$ and $v$ arrive, we delete $(t^{(0)}_i,t^{(1)}_i)$
and $\Edge{b^{(0)}_{i}}{b^{(1)}_{i}}$ iff $u_i=0$ and
delete $\Edge{r^{(0)}_j}{r^{(1)}_j}$ iff $v_j=0$.
See \Figure{even-4-cycle} for an example with $k=4$.

We have $u^\top Mv= 1$ iff there is a cycle of length $k$ in $H$:
By \Observation{h-tree},
every cycle must contain at least one vertex in $T^{(1)}$, $B^{(1)}$,
or $R^{(1)}$.
By \Observation{h-cycle-side-long}, every cycle that contains a vertex
in $T^{(1)} \cup B^{(1)}$, but not in $R^{(1)}$ or vice versa has length at
least $2(\lfloor\frac{k}{2}\rfloor -1) + 4 \geq k+1$.
In total, we need to do $2n_1+n_2=\bigO(\sqrt{m})$ updates and $1$ query.
\end{proof}

\begin{figure}[tb]
\centering
\includegraphics[width=.3\textwidth]{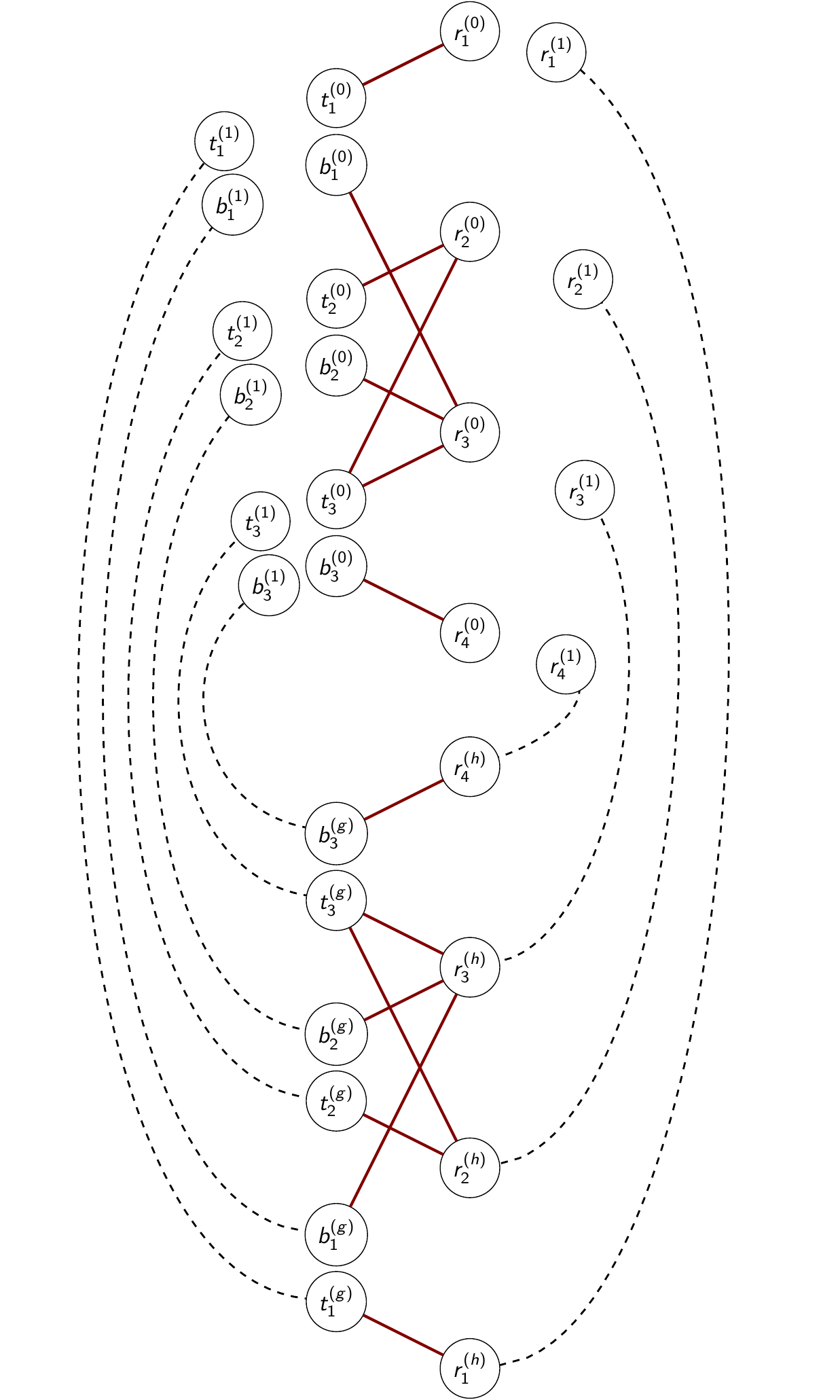}
\caption{Construction of $H_{M,g,h}$ for
$M = \Bigl(\begin{smallmatrix}
1 & 0 & 1 & 0 \\
0 & 1 & 1 & 0 \\
0 & 1 & 1 & 1 \\
\end{smallmatrix}\Bigr)$.}\FigLabel{lb-construct-H}
\end{figure}

\begin{figure}[tb]
\centering
\includegraphics[width=.3\textwidth]{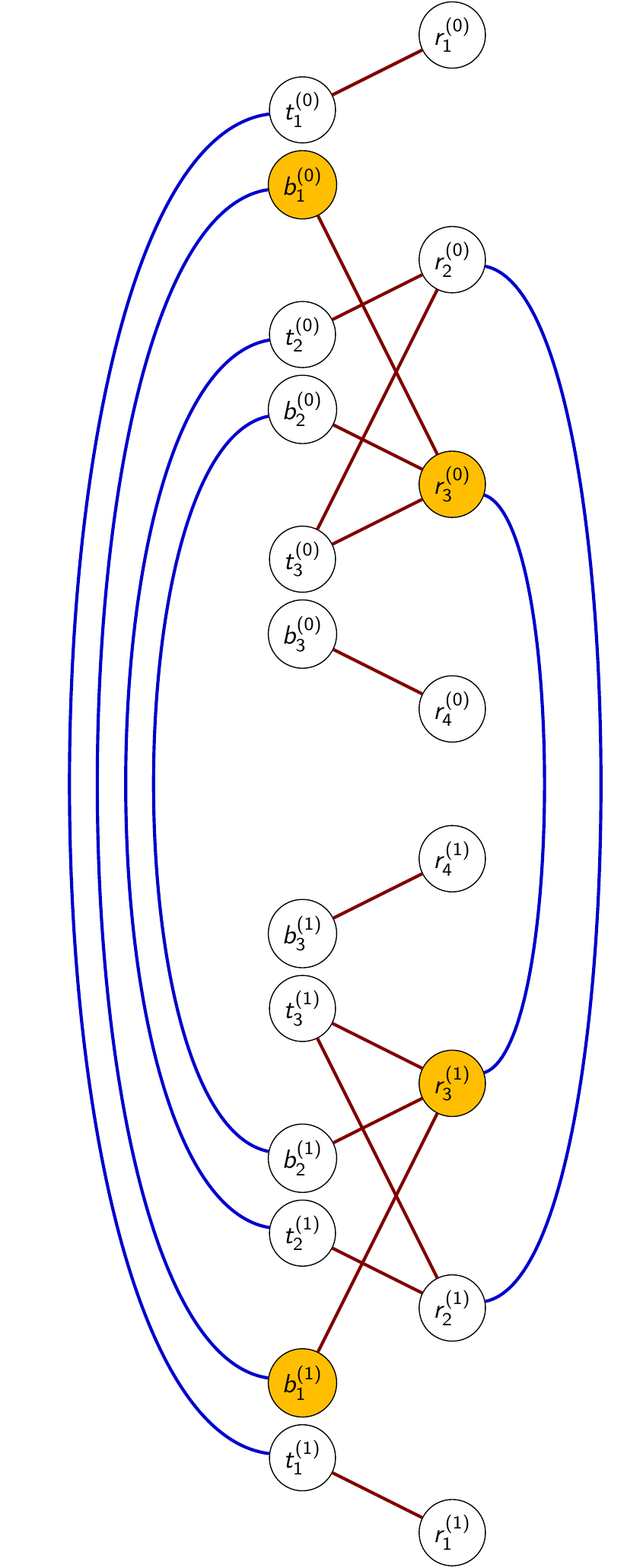}
\caption{Construction of $H$ for $4$-cycle detection with
$M = \Bigl(\begin{smallmatrix}
1 & 0 & 1 & 0 \\
0 & 1 & 1 & 0 \\
0 & 1 & 1 & 1 \\
\end{smallmatrix}\Bigr)$,
$u^\top = \begin{pmatrix}1&1&0 \end{pmatrix}$, and
$v^\top = \begin{pmatrix}0&1&1&0 \end{pmatrix}$.
The orange vertices $b^{(0)}_1, r^{(0)}_3, r^{(1)}_3, b^{(1)}_1$ form a $4$-cycle.}%
\label{fig:even-4-cycle}
\end{figure}

\end{document}